\documentclass[reqno,12pt,a4paper]{amsart}

\usepackage{mathrsfs}
\usepackage{amssymb}
\usepackage{amsfonts}
\usepackage{latexsym}
\usepackage{amsthm}
\usepackage{mathtools}
\DeclareMathOperator*{\slim}{s-lim}

\usepackage{enumitem}
\setlist{nosep, leftmargin=*}
\usepackage{titletoc}
\usepackage{pict2e}

\numberwithin{equation}{section}

\usepackage{graphicx}
\usepackage{tikz}
\begin{document}

%%%%%%%%%% Some definitions %%%%%%%%%%

%%%%%%%% Equations, theorems %%%%%%%%%
\renewcommand{\theequation}{\arabic{section}.\arabic{equation}}
\theoremstyle{plain}
\newtheorem{theorem}{\bf Theorem}[section]
\newtheorem{lemma}[theorem]{\bf Lemma}
\newtheorem{corollary}[theorem]{\bf Corollary}
\newtheorem{proposition}[theorem]{\bf Proposition}
\newtheorem{definition}[theorem]{\bf Definition}
\newtheorem*{definition*}{\bf Definition}
\newtheorem*{example}{\bf Example}
\newtheorem*{theorem*}{\bf Theorem}
\theoremstyle{remark}
\newtheorem*{remark}{\bf Remark}

%%%%% Alphabet %%%%%
\def\a{\alpha}  \def\cA{{\mathcal A}}     \def\bA{{\bf A}}  \def\mA{{\mathscr A}}
\def\b{\beta}   \def\cB{{\mathcal B}}     \def\bB{{\bf B}}  \def\mB{{\mathscr B}}
\def\g{\gamma}  \def\cC{{\mathcal C}}     \def\bC{{\bf C}}  \def\mC{{\mathscr C}}
\def\G{\Gamma}  \def\cD{{\mathcal D}}     \def\bD{{\bf D}}  \def\mD{{\mathscr D}}
\def\d{\delta}  \def\cE{{\mathcal E}}     \def\bE{{\bf E}}  \def\mE{{\mathscr E}}
\def\D{\Delta}  \def\cF{{\mathcal F}}     \def\bF{{\bf F}}  \def\mF{{\mathscr F}}
\def\c{\chi}    \def\cG{{\mathcal G}}     \def\bG{{\bf G}}  \def\mG{{\mathscr G}}
\def\z{\zeta}   \def\cH{{\mathcal H}}     \def\bH{{\bf H}}  \def\mH{{\mathscr H}}
\def\e{\eta}    \def\cI{{\mathcal I}}     \def\bI{{\bf I}}  \def\mI{{\mathscr I}}
\def\p{\psi}    \def\cJ{{\mathcal J}}     \def\bJ{{\bf J}}  \def\mJ{{\mathscr J}}
\def\vT{\Theta} \def\cK{{\mathcal K}}     \def\bK{{\bf K}}  \def\mK{{\mathscr K}}
\def\k{\kappa}  \def\cL{{\mathcal L}}     \def\bL{{\bf L}}  \def\mL{{\mathscr L}}
\def\l{\lambda} \def\cM{{\mathcal M}}     \def\bM{{\bf M}}  \def\mM{{\mathscr M}}
\def\L{\Lambda} \def\cN{{\mathcal N}}     \def\bN{{\bf N}}  \def\mN{{\mathscr N}}
\def\m{\mu}     \def\cO{{\mathcal O}}     \def\bO{{\bf O}}  \def\mO{{\mathscr O}}
\def\n{\nu}     \def\cP{{\mathcal P}}     \def\bP{{\bf P}}  \def\mP{{\mathscr P}}
\def\r{\varrho} \def\cQ{{\mathcal Q}}     \def\bQ{{\bf Q}}  \def\mQ{{\mathscr Q}}
\def\s{\sigma}  \def\cR{{\mathcal R}}     \def\bR{{\bf R}}  \def\mR{{\mathscr R}}
\def\S{\Sigma}  \def\cS{{\mathcal S}}     \def\bS{{\bf S}}  \def\mS{{\mathscr S}}
\def\t{\tau}    \def\cT{{\mathcal T}}     \def\bT{{\bf T}}  \def\mT{{\mathscr T}}
\def\f{\phi}    \def\cU{{\mathcal U}}     \def\bU{{\bf U}}  \def\mU{{\mathscr U}}
\def\F{\Phi}    \def\cV{{\mathcal V}}     \def\bV{{\bf V}}  \def\mV{{\mathscr V}}
\def\P{\Psi}    \def\cW{{\mathcal W}}     \def\bW{{\bf W}}  \def\mW{{\mathscr W}}
\def\o{\omega}  \def\cX{{\mathcal X}}     \def\bX{{\bf X}}  \def\mX{{\mathscr X}}
\def\x{\xi}     \def\cY{{\mathcal Y}}     \def\bY{{\bf Y}}  \def\mY{{\mathscr Y}}
\def\X{\Xi}     \def\cZ{{\mathcal Z}}     \def\bZ{{\bf Z}}  \def\mZ{{\mathscr Z}}
\def\O{\Omega}

\newcommand{\mc}{\mathscr {c}}

\newcommand{\gA}{\mathfrak{A}}          \newcommand{\ga}{\mathfrak{a}}
\newcommand{\gB}{\mathfrak{B}}          \newcommand{\gb}{\mathfrak{b}}
\newcommand{\gC}{\mathfrak{C}}          \newcommand{\gc}{\mathfrak{c}}
\newcommand{\gD}{\mathfrak{D}}          \newcommand{\gd}{\mathfrak{d}}
\newcommand{\gE}{\mathfrak{E}}
\newcommand{\gF}{\mathfrak{F}}           \newcommand{\gf}{\mathfrak{f}}
\newcommand{\gG}{\mathfrak{G}}           \newcommand{\Gg}{\mathfrak{g}}
\newcommand{\gH}{\mathfrak{H}}           \newcommand{\gh}{\mathfrak{h}}
\newcommand{\gI}{\mathfrak{I}}           \newcommand{\gi}{\mathfrak{i}}
\newcommand{\gJ}{\mathfrak{J}}           \newcommand{\gj}{\mathfrak{j}}
\newcommand{\gK}{\mathfrak{K}}            \newcommand{\gk}{\mathfrak{k}}
\newcommand{\gL}{\mathfrak{L}}            \newcommand{\gl}{\mathfrak{l}}
\newcommand{\gM}{\mathfrak{M}}            \newcommand{\gm}{\mathfrak{m}}
\newcommand{\gN}{\mathfrak{N}}            \newcommand{\gn}{\mathfrak{n}}
\newcommand{\gO}{\mathfrak{O}}
\newcommand{\gP}{\mathfrak{P}}             \newcommand{\gp}{\mathfrak{p}}
\newcommand{\gQ}{\mathfrak{Q}}             \newcommand{\gq}{\mathfrak{q}}
\newcommand{\gR}{\mathfrak{R}}             \newcommand{\gr}{\mathfrak{r}}
\newcommand{\gS}{\mathfrak{S}}              \newcommand{\gs}{\mathfrak{s}}
\newcommand{\gT}{\mathfrak{T}}             \newcommand{\gt}{\mathfrak{t}}
\newcommand{\gU}{\mathfrak{U}}             \newcommand{\gu}{\mathfrak{u}}
\newcommand{\gV}{\mathfrak{V}}             \newcommand{\gv}{\mathfrak{v}}
\newcommand{\gW}{\mathfrak{W}}             \newcommand{\gw}{\mathfrak{w}}
\newcommand{\gX}{\mathfrak{X}}               \newcommand{\gx}{\mathfrak{x}}
\newcommand{\gY}{\mathfrak{Y}}              \newcommand{\gy}{\mathfrak{y}}
\newcommand{\gZ}{\mathfrak{Z}}             \newcommand{\gz}{\mathfrak{z}}

\def\ve{\varepsilon}   \def\vt{\vartheta}    \def\vp{\varphi}    \def\vk{\varkappa}

\def\A{{\mathbb A}} \def\B{{\mathbb B}} \def\C{{\mathbb C}}
\def\dD{{\mathbb D}} \def\E{{\mathbb E}} \def\dF{{\mathbb F}} \def\dG{{\mathbb G}}
\def\H{{\mathbb H}}\def\I{{\mathbb I}} \def\J{{\mathbb J}} \def\K{{\mathbb K}} \def\dL{{\mathbb L}}
\def\M{{\mathbb M}} \def\N{{\mathbb N}} \def\O{{\mathbb O}} \def\dP{{\mathbb P}} \def\R{{\mathbb R}}
\def\dQ{{\mathbb Q}} \def\S{{\mathbb S}} \def\T{{\mathbb T}} \def\U{{\mathbb U}} \def\V{{\mathbb V}}
\def\W{{\mathbb W}} \def\X{{\mathbb X}} \def\Y{{\mathbb Y}} \def\Z{{\mathbb Z}}

\newcommand{\1}{\mathbbm 1}
\newcommand{\dd}    {\, \mathrm d}

%%%%%%%%%%%%%%%%%%%%%%%

%%%%% Arrows %%%%%

\def\la{\leftarrow}              \def\ra{\rightarrow}            \def\Ra{\Rightarrow}
\def\ua{\uparrow}                \def\da{\downarrow}
\def\lra{\leftrightarrow}        \def\Lra{\Leftrightarrow}

%%%%% Typography %%%%%

\def\lt{\biggl}                  \def\rt{\biggr}
\def\ol{\overline}               \def\wt{\widetilde}
\def\no{\noindent}

%%%%% Math signs %%%%%

\let\ge\geqslant                 \let\le\leqslant
\def\lan{\langle}                \def\ran{\rangle}
\def\/{\over}                    \def\iy{\infty}
\def\sm{\setminus}               \def\es{\emptyset}
\def\ss{\subset}                 \def\ts{\times}
\def\pa{\partial}                \def\os{\oplus}
\def\om{\ominus}                 \def\ev{\equiv}
\def\iint{\int\!\!\!\int}        \def\iintt{\mathop{\int\!\!\int\!\!\dots\!\!\int}\limits}
\def\el2{\ell^{\,2}}             \def\1{1\!\!1}
\def\sh{\sharp}
\def\wh{\widehat}
\def\ds{\dotplus}
%%%%% Math operations %%%%%

\def\all{\mathop{\mathrm{all}}\nolimits}
\def\where{\mathop{\mathrm{where}}\nolimits}
\def\as{\mathop{\mathrm{as}}\nolimits}
\def\Area{\mathop{\mathrm{Area}}\nolimits}
\def\arg{\mathop{\mathrm{arg}}\nolimits}
\def\adj{\mathop{\mathrm{adj}}\nolimits}
\def\const{\mathop{\mathrm{const}}\nolimits}
\def\det{\mathop{\mathrm{det}}\nolimits}
\def\diag{\mathop{\mathrm{diag}}\nolimits}
\def\diam{\mathop{\mathrm{diam}}\nolimits}
\def\dim{\mathop{\mathrm{dim}}\nolimits}
\def\dist{\mathop{\mathrm{dist}}\nolimits}
\def\Im{\mathop{\mathrm{Im}}\nolimits}
\def\Iso{\mathop{\mathrm{Iso}}\nolimits}
\def\Ker{\mathop{\mathrm{Ker}}\nolimits}
\def\Lip{\mathop{\mathrm{Lip}}\nolimits}
\def\rank{\mathop{\mathrm{rank}}\limits}
\def\Ran{\mathop{\mathrm{Ran}}\nolimits}
\def\Re{\mathop{\mathrm{Re}}\nolimits}
\def\Res{\mathop{\mathrm{Res}}\nolimits}
\def\res{\mathop{\mathrm{res}}\limits}
\def\sign{\mathop{\mathrm{sign}}\nolimits}
\def\supp{\mathop{\mathrm{supp}}\nolimits}
\def\Tr{\mathop{\mathrm{Tr}}\nolimits}
\def\AC{\mathop{\rm AC}\nolimits}
\def\BBox{\hspace{1mm}\vrule height6pt width5.5pt depth0pt \hspace{6pt}}

%%%%%%%%%%%%% specialities %%%%%%%%%%%%%%

\newcommand\nh[2]{\widehat{#1}\vphantom{#1}^{(#2)}}
%{{\mathop{#1}\limits^\wedge}\vphantom{#1}^{(#2)}}
\def\dia{\diamond}

\def\Oplus{\bigoplus\nolimits}

%%%%%%%%%%% End of definitions %%%%%%%%%%

%\DeclareMathOperator{\Ric}    {Ric}
%\newcommand{\dd}    {\, \mathrm d}    % not optimal: no \, if at beginning

%%%%% OLD OLD OLD

\def\qqq{\qquad}
\def\qq{\quad}
\let\ge\geqslant
\let\le\leqslant
\let\geq\geqslant
\let\leq\leqslant

\newcommand{\ca}{\begin{cases}}
\newcommand{\ac}{\end{cases}}
\newcommand{\ma}{\begin{pmatrix}}
\newcommand{\am}{\end{pmatrix}}
\renewcommand{\[}{\begin{equation}}
\renewcommand{\]}{\end{equation}}
\def\bu{\bullet}

\title[{}]{Inverse resonance scattering for massless Dirac operators on the real line}

\date{\today}

\author[Evgeny Korotyaev]{Evgeny Korotyaev}
\address{Department of Analysis, Saint Petersburg State University, Universitetskaya nab. 7/9,
St. Petersburg, 199034, Russia, \ korotyaev@gmail.com, \ e.korotyaev@spbu.ru}
\author[Dmitrii Mokeev]{Dmitrii Mokeev}
\address{Saint Petersburg State University, Universitetskaya nab. 7/9,
St. Petersburg, 199034, Russia, \ mokeev.ds@yandex.ru}


\begin{abstract}
    We consider massless Dirac operators on the real line with compactly supported potentials.
    We solve two inverse problems (including characterization): in terms of zeros of reflection coefficient
    and in terms of poles of reflection coefficients (i.e. resonances). We prove that a potential is
    uniquely determined by zeros of reflection coefficients and there exist distinct potentials
    with the same resonances. We describe the set of "isoresonance potentials".
    Moreover, we prove the following:

    1) a zero of the reflection coefficient can be arbitrarily shifted, such that we obtain
    the sequence of zeros of the reflection coefficient for an other compactly supported potential,

    2) the forbidden domain for resonances is estimated,

    3) asymptotics of resonances counting function is determined,

    4) these results are applied to canonical systems.
\end{abstract}

\maketitle

\tableofcontents

\section{Introduction} \label{p0}
We consider an inverse problem for Dirac operators on the real line
with compactly supported potentials. Such operators have many
physical and mathematical applications. These Dirac operators are also
known as Zakharov-Shabat (or AKNS) systems, which were
used by Zakharov and Shabat \cite{ZS71} to study nonlinear
Schr{\"o}dinger equation (see also \cite{APT04, DEGM82, FT07}). In our paper,
we consider the self-adjoint Dirac operator $H$ on $L^2(\R,\C^2)$
given~by
\[ \label{intro:operator}
    H y = -i \s_3 y' + i \s_3 Q y,\qq y = \ma y_1 \\ y_2 \am,\qq \s_3 = \ma 1 & 0 \\ 0 & -1 \am.
\]
The potential $Q$ has the following form
\[ \label{intro:potential}
    Q = \ma 0 & q \\ \overline{q} & 0 \am,\qq q \in \cP,
\]
where the class $\cP$ is defined for some $\g > 0$ fixed throughout
this paper by
\begin{definition*}
    $\cP = \cP_{\g}$ is a set of all functions $q \in L^2(\R)$ such that the convex hull of $\supp q$
    equals $[0,\g]$.
\end{definition*}
Recall that
$\s(H) = \s_{ac}(H) = \R$ (see e.g. \cite{LS91}).
We introduce the $2 \times 2$ matrix-valued Jost solutions $f^{\pm}(x,k) =  \left(
\begin{smallmatrix} f^{\pm}_{11} & f^{\pm}_{12} \\ f^{\pm}_{21} & f^{\pm}_{22} \end{smallmatrix} \right) (x,k)$
of the Dirac equation
\[ \label{intro:equation}
    (f^{\pm})'(x,k) = Q(x) f^{\pm}(x,k) + i k \s_3 f^{\pm}(x,k),\qq (x,k) \in \R \ts \C,
\]
which satisfy the standard condition for compactly supported potentials:
$$
    \begin{aligned}
        f^{+}(x,k) = e^{i k x \s_3},\qq \forall \qq x \geq \g,\\
        f^{-}(x,k) = e^{i k x \s_3},\qq \forall \qq x \leq 0.
    \end{aligned}
$$
Since equation (\ref{intro:equation}) has exactly one linear independent solution, it follows that
for any $k \in \C$, there exists a unique $2 \times 2$ \textit{transition matrix} $A(k)$ such that
$$
    f^{+}(x,k) = f^{-}(x,k)A(k),\qq x \in \R.
$$
The transition matrix $A$ has the form
$$
    A = \ma a & b_* \\ b & a_* \am,\qq a a_* - b b_* = 1,
$$
where we used the notation $g_*(k) = \ol{g(\ol{k})}$, $k \in \C$.
It is well-known that $a$ and $b$ are entire,
$a(k) \neq 0$ for any $k \in \ol \C_+ $ and it has zeros in $\C_-$, which are called
\textit{resonances} and they are also zeros of the Fredholm
determinant and poles of the resolvent of the operator $H$ (see e.g.
\cite{IK14a}). Note that the zeros of $b$ and $a$ do not coincide.
Let $H_o = -i \s_3$ be the free Dirac operator on $L^2(\R,\C^2)$.
Then the scattering matrix $S$ for the pair $H$, $H_o$ has the following form
\[ \label{p2e3}
    S(k) = \frac{1}{a(k)} \ma 1 & -\ol{b(k)} \\ b(k) & 1\am,\qq k \in \R.
\]
Here $1/a$ is a transmission coefficient and $r_+ = -\ol b/a$ (or $r_- = b/a$) is a right
(or left) reflection coefficient. The matrix-valued function $S$ admits a meromorphic continuation from $\R$ onto
$\C$, since $a$ and $b$ are entire. Poles of $S$ are resonances and zeros of the reflection
coefficients $r_{\pm}$ coincide with zeros of $b$ or $b_*$.
We sometimes write $a(\cdot,q)$, $b(\cdot,q)$, $\ldots$ instead of $a(\cdot)$, $b(\cdot)$, $\ldots$,
when several potentials are being dealt with.

Our main goal is to solve inverse problems for the Dirac operator
$H$ with different spectral data: \textit{the coefficients of the transition matrix} $a$ and $b$,
\textit{the coefficients of the scattering matrix} $r_{\pm}$, \textit{the zeros} of $b$,
and \textit{the resonances}. In general, an inverse problem is to determine the potential by some
data, and it consists at least of the four parts:
\begin{enumerate}[label={(\roman*)}]
    \item {\it Uniqueness.} Do data uniquely determine the potential?
    \item {\it Reconstruction.} Give an algorithm to recover the potential by data.
    \item {\it Characterization.} Give necessary and sufficient conditions that data correspond
    to a potential.
    \item {\it Continuity.} Is a potential a continuous function of data and how can
    data be changed so that they remain data for some potential?
\end{enumerate}

Firstly, we consider the inverse problem in terms of the reflection coefficients $r_{\pm}$,
the coefficient of transition matrix $b$ or the zeros of $b$, which coincide with zeros of $r_{\pm}$
as was noted above. For these data, we obtain:
\begin{enumerate}[label={(\roman*)}]
    \item Each of these data determine a potential uniquely.
    \item We solve the reconstruction and characterization problem for these data.
    \item We also solve these problems for even, odd or real-valued potentials.
    \item We solve the continuity problem in terms of $r_{\pm}$ and $b$ and
    we partially solve this problem in terms of zeros of $b$.
    Namely, we show that if a zero of $b$ is arbitrarily shifted, then we obtain a
    coefficient $\tilde{b}$ for some potential from $\cP$. We also prove that a potential
    continuously depends on one zero of $b$, where its other zeros are
    fixed.
\end{enumerate}

Secondly, we consider the inverse problem in terms of the coefficient $a$ of a transition matrix or
in terms of its zeros (the resonances). These data do not uniquely determine a potential from $\cP$.
In this case, we obtain:
\begin{enumerate}[label={(\roman*)}]
    \item We prove the uniqueness by adding to data a sequence
    $(\xi_n)_{n \geq 0}$, where $\xi_0 = e^{i \vp}$ for some $\vp \in \R$ and
    $\xi_n \in \{-1,0,1\}$, $n \geq 1$. In generic case, $\xi_0$ is a
    phase multiplier of $b$ and $\xi_n = \sign \Im z_n$, where $(z_n)_{n \geq 1}$ is a sequence
    of zeros of $b$ in $\C \sm \{0\}$.
    \item We solve the characterization problem for such extended data and describe isoresonance sets, i.e. the sets of potentials which have the same resonances.
    \item We also solve these problems for even, odd or real-valued potentials.
    \item Finally, we consider the stability problem for resonances and we solve it in some special cases.
    Note that the resonances do not completely determine a potential and they are not free parameters
    of the Dirac operator, i.e. we can not arbitrarily move a resonance. On the other hand,
    the zeros of $b$ completely determine a potential and they are free parameters.
\end{enumerate}

The coefficients $a$ and $b$ are important in study of the nonlinear
Schr{\"o}dinger equation (NLS). Namely, $\log |a(k)|$ and $\arg b(k)$, $k \in \R$,
are action-angle variables for the NLS equation (see e.g. p.~230 in \cite{FT07}).
Using the analytical properties of $a$ and $b$, we obtain
the representations of action-angle variables in terms of resonances and zeros of $b$.

Note also that the Dirac operators can be rewritten as canonical systems (see
e.g. p.~389 in \cite{GK67}). Using this relation, we describe the class of canonical systems,
which are unitary equivalent to the Dirac operators. Then we introduce the scattering matrix and
solve the inverse scattering problem for such canonical systems.

In our paper, we use methods from the paper \cite{K05}, where
the similar problems for the Schr{\"o}dinger operators on the line were
considered. However, there exist differences between Dirac and Schr{\"o}dinger cases,
which require an adaptation of the proofs. We describe
the main differences between Dirac and Schr{\"o}dinger cases:

\begin{enumerate}[label={(\roman*)}]
    \item In general, the resonances of the Dirac operators are not symmetric with respect to the imaginary
    line.

    \item Roughly speaking, the spectral problem for Dirac operators
    corresponds to spectral problem for the Schr{\"o}dinger operators with
    distributions.

    \item The second term in the asymptotic expansion of the Jost solutions of the Dirac operators
    decrease more slowly as spectral parameter goes to infinity. Maybe it is the
    main point.
\end{enumerate}

In our paper, we also use the solution of the inverse scattering problem for
the Dirac operators with not necessarily compactly supported potentials.
This problem has been widely studied, see \cite{APT04, FT07, FHMP09} and references therein.

There are a lot papers about resonances in the different setting,
see articles \cite{F97, H99, K04a, S00, Z87} and the book
 \cite{DZ19} and the references therein. The inverse resonance problem for
Schr{\"o}dinger operators with compactly supported potentials was
solved in \cite{K05} for the case of the real line and in
\cite{K04a} for the case of the half line. In these papers, the
uniqueness, reconstruction, and characterization problems were
solved, see also Zworski \cite{Z02}, Brown-Knowles-Weikard
\cite{BKW03} concerning the uniqueness.
Moreover, there are other results about
perturbations of the following model (unperturbed) potentials by
compactly supported potentials: step potentials \cite{C06}, periodic
potentials \cite{K11h}, and linear potentials (corresponding to
one-dimensional Stark operators) \cite{K17}.

In the theory of resonances, one of the basic result is the
asymptotics of the counting function of resonances, which is an
analogue of the Weyl law for eigenvalues. For Schr{\"o}dinger
operators on the real line with compactly supported potentials, such
result was first obtained by Zworski in \cite{Z87}. The "local
resonance" stability problems were considered in \cite{K04a, K04b, K05, MSW10}
and results about the Carleson measures for resonances
were obtained in \cite{K16}.

In our paper, we consider the inverse resonance problem for Dirac
operators on the real line. As far as we know, this problem has not
been studied enough. Now, we shortly discuss the known results on
the resonances of one-dimensional Dirac operators. Global estimates
of resonances for the massless Dirac operators on the real line were
obtained in \cite{K14}. Resonances for Dirac operators was also
studied in \cite{IK14b} for the massive Dirac operators on the
half-line and in \cite{IK14a} for the massless Dirac operators on
the real line for smooth compactly supported potentials $q' \in L^1(\R)$.
If $q' \in L^1(\R)$, then the second term in the asymptotic
expansion of the Jost solutions of the Dirac operators decrease as in case of the Schr{\"o}dinger
operator as spectral parameter goes to infinity. In these papers,
the following results were obtained:
\begin{enumerate}[label={(\roman*)}]
    \item asymptotics of counting function of the resonances;
    \item estimates on the resonances and the forbidden domain;
    \item the trace formula in terms of resonances for the massless case.
\end{enumerate}
In \cite{IK15}, the radial Dirac operator was considered.
The inverse resonance problem for the massless Dirac operator on the half-line
with compactly supported potentials
was solved in \cite{KM20}. Note that the problem on the half-line is simpler
than the problem on the real line and the main differences are
\begin{enumerate}[label={(\roman*)}]
    \item the scattering matrix depends on two coefficients $a$ and $b$;
    \item the potential are not uniquely determined by resonances;
    \item the resonances are not free parameters.
\end{enumerate}
There is a
number of papers dealing with other related problems for the
one-dimensional Dirac operators, for instance, the resonances for
Dirac fields in black holes was described, see e.g., \cite{I18}.

Recall that Dirac operators can be rewritten as canonical systems.
For these systems, the inverse problem can be solved in terms of
de Branges spaces (see \cite{dB, R14}).
There exist many papers devoted to de Branges spaces and canonical
systems. In particular, they are used in the inverse spectral theory
of Schr{\"o}dinger and Dirac operators (see e.g. \cite{R02}). It is
well-known that there exist the connection between Jost solutions
and de Branges spaces. In paper \cite{KM20}, the characterization
of de Brange spaces associated with the Dirac operators was given. Similar
characterization in case of the Schr{\"o}dinger operators was given
in \cite{BBP} (see also \cite{P}).

\section{Main results} \label{p1}
We introduce the Fourier transform $\cF$ on $L^2(\R)$ by
$$
    (\cF g)(k) = \int_{\R} g (s) e^{2iks} ds,\qq k \in \R.
$$
Then its inverse $\cF^{-1}$ on $L^2(\R)$ is given by
$$
    \cF^{-1} g(s) = \frac{1}{\pi} \int_{\R} g(k) e^{-2iks} dk,\qq s \in \R.
$$
We will use the notation $\hat g = \cF^{-1} g$.
Now, we introduce the classes of scattering data associated with the zeros of the reflection coefficients.

\begin{definition*}
    $\cB = \cB_{\g}$ is a metric space of all entire functions $b$ such that
    $\hat b \in \cP_{\g}$ equipped with the metric
    \[ \label{p2e10}
        \rho_{\cB}(b_1,b_2) = \| \hat b_1 - \hat b_2 \|_{L^2(0,\g)},\qq b_1,b_2 \in \cB.
    \]
\end{definition*}
We recall well-known facts about entire functions. An entire function $f(z)$ is said to be of
\textit{exponential type} if there exist constants $\t,C > 0$ such that $|f(z)| \leq C e^{\t |z|}$,
$z \in \C$. We introduce the Cartwright class of entire functions $\cE_{Cart}$ by
\begin{definition*}
    $\cE_{Cart}$ is a class of entire functions $f$ of exponential type such that
    $$
        \int_{\R} \frac{\log(1+|f(k)|)dk}{1 + k^2} < \iy,\qq \t_+(f) = 0,\qq \t_-(f) = 2\g,
    $$
    where $\t_{\pm}(f) = \lim \sup_{y \to +\iy} \frac{\log |f(\pm i y)|}{y}$.
\end{definition*}
\begin{remark}
    It follows from the Paley-Wiener Theorem (see e.g. p.30 in \cite{Koo98}) that $\cB \ss \cE_{Cart}$.
\end{remark}
Note that the metric space $\cB$ is not complete.
We also equip the class $\cP$ with the metric $\rho_{\cP}$ given by
\[ \label{p1e16}
    \rho_{\cP}(q_1,q_2) = \|q_1 - q_2\|_{L^2(0,\g)},\qq q_1,q_2 \in \cP.
\]
Thus, $\cP$ is a metric space, which is also not complete.
Now, we present our first result.

\begin{theorem} \label{t1}
        The mapping $q \mapsto b(\cdot,q)$ is a homeomorphism between $\cP$ and $\cB$.
\end{theorem}
\begin{remark}
    1) In Corollary \ref{p9c1}, we give the characterization of $b$ for even, odd or real-valued potentials.

    2) Similar results for the Schr{\"o}dinger operators were obtained in \cite{K05}.
\end{remark}

We also introduce the class of scattering data associated with the resonances, i.e.,
with the poles of the reflection coefficients. We need the following notation.

\begin{definition*}
    $\cA = \cA_{\g}$ is the set of all entire functions $a$ such that:
    \begin{enumerate}[label={\roman*)}]
        \item $a(k) \neq 0$ for any $k \in \C_+$;
        \item $|a(k)| \geq 1$ for any $k \in \R$;
        \item $|a|^2 - 1 \in L^1(\R)$;
        \item $a = 1 + \cF h$ for some $h \in \cP_{\g}$.
    \end{enumerate}
\end{definition*}
\begin{remark}
    By the Paley-Wiener Theorem, $\cA \ss \cE_{Cart}$.
\end{remark}
Below, we show that the coefficient $a$ does not uniquely determine a potential.
Thus, we need additional data from the coefficient $b$. Let $b \in \cB$ and let
$(z_n)_{n \geq 1}$, be the zeros of $b$ in $\C \sm \{ 0 \}$ counted with multiplicity and arranged
that $0 < |z_1| \leq |z_2| \leq \ldots$. Let also $p \geq 0$ be the multiplicity of the zero
$k = 0$ of $b(k)$. Then we introduce a sequence $\xi(b) = (\xi_n)_{n \geq 0}$ by
\[ \label{p1e18}
    \xi_0 = \frac{b^{(p)}(0)}{|b^{(p)}(0)|},\qq
    \xi_n = \sign \Im z_n =
    \begin{cases}
        1,& \Im z_n > 0\\
        0,& \Im z_n  = 0\\
        -1,& \Im z_n < 0
    \end{cases},\qq n \geq 1.
\]
Using $\xi(b)$, we can parametrize the space of solutions $b \in \cB$ of the equation
$b b_* = B$ for some exponential type function $B$ such that $B(z) \geq 0$ for any $z \in \R$.
In this case, the zeros of $B$ are symmetric with respect to the real line and then some zeros are
zeros of $b$ and the other are zeros of $b_*$. Thus, for such function $B$, we introduce
the following set
$$
    \Xi(B) = \{ \, \xi(b) \, \mid \, b \in \cB,\, b b_* = B \, \}.
$$
Now, we give our second main result.
\begin{theorem} \label{t2}
    The mapping $q \mapsto (a, \xi)$ is a bijection between $\cP$ and $\cA \ts \Xi(B)$,
    where $a = a(\cdot,q)$, $\xi = \xi(b(\cdot,q))$, and $B = aa_*-1$.
\end{theorem}
\begin{remark}
    1) Due to Lemma \ref{hll6}, for each $a \in \cA$, there exists at least one $b \in \cB$
    such that $aa_* - bb_* = 1$.
    It follows that $\Xi(aa_* - 1) \neq \es$ for any $a \in \cA$.

    2) We can give another characterization of solutions $b \in \cB$
    of the equation $aa_* - bb_* = 1$ for some fixed $a \in \cA$. Let $B = aa_*-1$ and let $z_o \in \C$
    be its zero with multiplicity $n_o \geq 1$. Below, we show that $\ol z_o$ is a zero of $B$ with
    multiplicity $n_o$. Let $b \in \cB$ and $bb_* = B$. Then, $z_o$ is a zero of $b$ with
    multiplicity $m_o$ and $\ol z_o$ is a zero of $b$ with multiplicity $n_o - m_o$.
    Thus, the solutions $b \in \cB$ are distinguished by the multiplicity of their zeros and
    by the constant $\xi_0(b)$.

    3) In Corollary \ref{p9c2}, we give the characterization of $a$ for even, odd or real-valued potentials.

    4) Similar result for the Schr{\"o}dinger operators was obtained in \cite{K05}.
\end{remark}

Using Theorems \ref{t1} and \ref{t2}, we give the characterization of the reflection coefficients
for the compactly supported potentials. We introduce the following classes.
\begin{definition*}
    $\cR^+ = \cR^+_{\g}$ (or $\cR^- = \cR^-_{\g}$) is a set of all meromorphic functions $r_+$ (or $r_-$),
    which have the form $r_+ = -\frac{\ol b}{a}$ (or $r_- = \frac{b}{a}$) for some
    $b \in \cB_{\g}$ and $a \in \cA_{\g}$ such that
    $$
        |a(k)|^2 - |b(k)|^2 = 1,\qq k \in \R.
    $$
\end{definition*}
Below, we prove that $\hat r_{\pm} \in L^2(\R) \cap L^1(\R)$ for any $r_{\pm} \in \cR^{\pm}$.
Thus, we introduce the metric
$$
    \rho_{\cR}(r_1,r_2) = \|\hat r_1 - \hat r_2\|_{L^2(\R)} + \|\hat r_1 - \hat r_2\|_{L^1(\R)},
    \qq r_1,r_2 \in \cR^{\pm}.
$$
Thus, $\cR^{\pm}$ are incomplete metric spaces equipped with the metric $\rho_{\cR}$.
\begin{theorem} \label{t11}
    The mappings $q \mapsto r_{\pm}(\cdot,q)$ are homeomorphisms between $\cP$ and $\cR^{\pm}$.
\end{theorem}

Above, we introduced the Cartwright class of entire functions $\cE_{Cart}$.
Such functions have some remarkable properties.
Let $f \in \cE_{Cart}$ and let $p$ be the multiplicity of zero of $f(k)$ at $k = 0$.
We denote by $(k_n)_{n \geq 1}$ zeros of $f$ in $\C \sm \{ 0 \}$
counted with multiplicity and arranged that $0 < |k_1| \leq |k_2| \leq \ldots$.
Then $f$ has the Hadamard factorization
\[ \label{p2e13}
    f(k) = C k^p e^{i \g k} \lim_{r \to +\iy}
        \prod_{|k_n| \leq r} \left(1 - \frac{k}{k_n}\right),\qq k \in \C,
\]
see, e.g., pp.127-130 in \cite{L96}, where the product converges uniformly
on compact subsets of $\C$ and
\[
    C = \frac{f^{(p)}(0)}{p!},\qq
    \sum_{n \geq 1} \frac{|\Im k_n|}{|k_n|^2} < +\iy,\qq
    \exists \lim_{r \to +\iy} \sum_{|k_n| \leq r} \frac{1}{k_n} \neq \iy.
\]
For any entire function $f$ and $(r,\d) \in \R_+ \ts [0,\frac{\pi}{2}]$,
we introduce the following counting functions
$$
    N_{\pm}(r,\d,f) = \# \{ \, k \in \C \, \mid \,
    f(k) = 0,\, |k| \leq r,\,\pm \Re k \geq 0,\, \d < |\arg k| < \pi-\d  \}.
$$
We need the Levinson's result about zeros of functions from
$\cE_{Cart}$, see, e.g., p. 58 in \cite{Koo98}.
\begin{theorem*}[Levinson]
    Let $f \in \cE_{Cart}$. Then for each $\d > 0$ we have
    \[ \label{p2e15}
        N_{\pm}(r,0,f) = \frac{\g}{\pi} r + o(r),\qq N_{\pm}(r,\d,f) = o(r),
    \]
    as $r \to +\iy$.
\end{theorem*}
Recall that $\cA,\cB \ss \cE_{Cart}$. Thus, using the properties of $\cE_{Cart}$, we get the following corollary.
\begin{corollary} \label{t3}
    \begin{enumerate}[label={\roman*)}]
        \item The potential $q \in \cP$ is uniquely determined by zeros of $b(\cdot,q) \in \cB$ and
        $b^{(p)}(0,q) \in \C \sm \{ 0 \}$, where $p$ is the multiplicity of zero of $b(k,q)$ at $k = 0$.
        Moreover, $b(\cdot,q)$ satisfies (\ref{p2e13} -- \ref{p2e15}).
        \item The potential $q \in \cP$ is uniquely determined by zeros of $a = a(\cdot,q) \in \cA$
        and by $\xi(b(\cdot,q)) \in \Xi(B)$, where $B = aa_*-1$.
        Moreover, $a(\cdot,q)$ satisfies (\ref{p2e13} -- \ref{p2e15}).
    \end{enumerate}
\end{corollary}
\begin{remark}
    The asymptotics of distribution of resonances of the Dirac operators on the real line was
    obtained in \cite{IK14a} for smooth compactly supported potentials $q' \in L^1(\R)$.
\end{remark}

It follows from Theorem \ref{t2} that for some $q_o \in \cP$ there exist distinct potentials $q \in \cP$
such that $a(\cdot,q_o) = a(\cdot,q)$. Moreover, the function $a \in \cA$ is uniquely determined
by its zeros. For any $q_o \in \cP$, we introduce an \textit{isoresonance set} in $\cP$ as follows
$$
    \Iso(q_o) = \{\, q \in \cP \, \mid \, a(\cdot,q) = a(\cdot,q_o)\, \}.
$$
Using Theorem \ref{t2} for fixed $a(\cdot,q_o)$,
we give the characterization of $\Iso(q_o)$ in terms of~$\Xi(B)$.
\begin{corollary} \label{c2}
    Let $q_o \in \cP$. Then the mapping $q \mapsto \xi$ is a bijection between $\Iso(q_o)$ and $\Xi(B)$,
    where $\xi = \xi(b(\cdot,q))$, $a = a(\cdot,q_o)$, and $B = aa_*-1$.
\end{corollary}

On the other hand, it follows from Theorem \ref{t1} that for each $q \in \cP$ there exists a unique
$b = b(\cdot,q) \in \cB$. Now, we describe how $b(\cdot,q)$ changes for $q \in \Iso(q_o)$.
\begin{theorem} \label{t14}
    Let $q^o \in \cP$ and let $(z_n)_{n \geq 1}$ be zeros of $b(\cdot,q^o)$ counted with multiplicity.
    Then $q \in \Iso(q^o)$ if and only if $b(\cdot,q) = b(\cdot,q^o) P(\cdot)$, where
    \[ \label{p1e3}
        P(k) = e^{i\a} \lim_{r \to +\iy} \prod_{z_n \in G,\, |z_n| < r}
        \left( 1 - \frac{k}{\ol z_n} \right)\left( 1 - \frac{k}{z_n} \right)^{-1},
    \]
    for any $k \in \C$ and for some $\a \in \R$ and non-real subsequence $G \ss (z_n)_{n \geq 1}$.
\end{theorem}
\begin{remark}
    It is easy to see that $P$ is a meromorphic function such that $P_* = \frac{1}{P}$,
    which yields that $|P(k)| = 1$ for each $k \in \R$ and then
    $$
        |b(k,q)| = |b(k,q^o)|,\qq |r_{\pm}(k,q)| = |r_{\pm}(k,q^o)|,\qq k \in \R,
    $$
    for each $q \in \Iso(q^o)$.
\end{remark}

Using the analytical properties of $a$ and $b$,
we obtain the representations of action-angle variables for the NLS equation
in terms of resonances and zeros of $b$. We give some known facts about NLS equation from \cite{FT07}.
Let $p(x,t)$
be a solution of the defocusing NLS equation
\[ \label{p1e22}
    i\frac{\partial p}{\partial t} = - \frac{\partial^2 p}{\partial t^2} + 2|p|^2 p
\]
satisfying the initial condition $p(\cdot,0) = q \in L^2(\R) \cap L^1(\R)$. For such equation,
there exists the Lax pair or the zero curvature representation and the Dirac equation is the
corresponding auxiliary linear problem. Thus, considering the solution $p(\cdot,t)$ of (\ref{p1e22})
as a potential of the Dirac operator $H$, we get $b(\cdot,p(\cdot,t))$, $a(\cdot,p(\cdot,t))$
and these functions have a simple dynamics:
\[ \label{p1e23}
    a(k,p(\cdot,t)) = a(k,q),\qq b(k,p(\cdot,t)) = e^{4itk^2} b(k,q),\qq (k,t) \in \R^2,
\]
see, e.g., p.~52 in \cite{FT07}.
Thus, solving the inverse problem for $b$ given by (\ref{p1e23}),
we obtain $p(x,t)$. Moreover, the NLS equation is a Hamiltonian system with the Hamiltonian
$$
    E = \int_{\R} \left( \left| \frac{\partial p}{\partial x}(x,t) \right|^2 + |p(x,t)|^4 \right) dx.
$$
Such system is completely integrable and there exist action variables $\r(k,t)$, $k \in \R$, and
angle variables $\phi(k,t)$, $k \in \R$, given by
$$
    \r(k,t) = \frac{1}{\pi} \log |a(k,q)|,\qq \phi(k,t) = 4k^2t + \arg b(k,q),\qq (k,t) \in \R^2,
$$
where $\log x \in \R$ and $\arg x = 0$ for any $x \in \R_+$ and $p(\cdot,0) = q$ (see, e.g., p.~230 in \cite{FT07}).
Now, we give the representation of these variables in terms of the resonances and zeros of $b$,
when $q \in \cP$.
\begin{theorem} \label{t12}
    Let $q \in \cP$ and let $a = a(\cdot,q)$, $b = b(\cdot,q)$, $\xi = \xi(b)$.
    Let also $(k_n)_{n \geq 1}$ be zeros of $a$ in $\C_-$ and
    let $(z_n)_{n \geq 1}$ be zeros of $b$ in $\C \sm \{0\}$. Then we have
    \begin{align}
        \log|a(k)| &= \log|a(0)| + \lim_{r \to +\iy}\sum_{|k_n| \leq r} \log \left| 1 - \frac{k}{k_n} \right|,\qq k \in \R, \label{p1e19}\\
        \arg b(k) &= \arg \xi_0 + \g k - \pi I(k) + \int_0^k w(s) ds,\qq k \in \R \label{p1e20},
    \end{align}
    where $\log x \in \R$ and $\arg x = 0$ for any $x \in \R_+$,
    \[ \label{p1e21}
        I(k) = \begin{cases}
            \# \{ z_n \in [0,k),\, n \geq 1 \}, & k \geq 0\\
            \# \{ z_n \in (k,0),\, n \geq 1 \}, & k < 0
        \end{cases}
        ,\qq w(k) = \sum_{n \geq 1} \xi_n \frac{|\Im z_n|}{|z_n - k|^2},\qq
        k \in \R,
    \]
    and the series in (\ref{p1e19}), (\ref{p1e21}) converge uniformly on compact subsets of $\R$.
\end{theorem}
\begin{remark}
    1) It follows from (\ref{p1e23}) that if $q \in \cP$, then $a(\cdot,p(\cdot,t))$ and
    $b(\cdot,p(\cdot,t))$ are entire functions for any $t \in \R$. Moreover,
    theirs zeros do not move when $t$ changes.

    2) If $q \in \cP$, then $e^{4itz^2}b(z,q)$ is not exponential type as function of $z$ for any $t \neq 0$.
    Thus, by Theorem \ref{t1}, $p(\cdot,t)$ does not have compact support for any $t \neq 0$.
\end{remark}

Now, we describe some properties of resonances and zeros of $b$. Firstly, we describe
a forbidden domain for resonances.
\begin{theorem} \label{t5}
    Let $q \in \cP$ and let $(k_n)_{n \geq 1}$ be its resonances. Let $\ve > 0$. Then there exists
    a constant $C = C(\ve,q) \geq 0$ such that the following inequality holds true for each $n \geq 1$:
    \[ \label{p1e1}
        2 \g \Im k_n \leq \ln \left( \ve + \frac{C}{|k_n|} \right).
    \]
    In particular, for any $A > 0$, there are only finitely many resonances in the strip
    \[ \label{p1e2}
        \{ \, k \in \C \, \mid \, \Im k \in (-A;0) \, \}.
    \]
\end{theorem}
\begin{remark}
    If $q' \in L^1(\R)$, then estimate (\ref{p1e1}) and the forbidden domain (\ref{p1e2})
    can be improved (see Theorem 1.3 in \cite{IK14a}).
\end{remark}
Secondly, we describe some automorphisms of the classes $\cA$ and $\cB$.
We show that the zeros of $b$ are free parameters and prove that $b$ continuously depends on a zero.
\begin{theorem} \label{t4}
    Let $q^o \in \cP$ and let $(z_n^o)_{n \geq 1}$ be zeros of $b(\cdot,q^o)$.
    Let $N = \# \{ z_j^o \mid |z_j^o| < r \}$ for some $r>1$.
    Let $z_j \in \C$, $|z_j|<r$, $j=1,\ldots,N$. Then there exist a unique $q \in \cP$ and
    a unique $p \in \cP$ such~that
    $$
        \begin{aligned}
            b(k,q) &= b(k,q^o)\prod_{j = 1}^N\frac{k-z_j}{k-z_j^o},\qq k \in \C,\\
            b(k,p) &= b(k,q^o)\prod_{j = 1}^N\frac{1}{k-z_j^o},\qq k \in \C.
        \end{aligned}
    $$
    In particular, any point on $\C$ can be a zero of $b$ with any
    multiplicity for some $q \in \cP$. Moreover, if each $z_j \to z_j^o$,
    $j=1,\ldots,N$, then we have
    $$
        \rho_{\cP}(q^o, q) \to 0,\qq
        \rho_{\cB}(b(\cdot,q^o), b(\cdot,q)) \to 0,\qq
        \rho_{\cR}(r_{\pm}(\cdot,q^o), r_{\pm}(\cdot,q)) \to 0.
    $$
\end{theorem}
\begin{remark}
    For Schr{\"o}dinger operator, similar results are obtained in \cite{K05}.
\end{remark}
Now, we consider the resonances. In general, the resonances of Dirac operators on
the real line are not free parameters. However, we prove that
the resonances can be shifted in some domain when they are symmetric with respect
to the imaginary line. We introduce the following subspace
$$
    \cA_{symm} = \{\, a \in \cA \, \mid \, a(k) = a_*(-k),\, k \in \C \,\}.
$$
\begin{theorem} \label{t8}
    Let $a(\cdot,q^o) \in \cA_{symm}$ for some $q^o \in \cP$ and let $a(k_o,q^o) = 0$
    for some $k_o \in \C_-$.
    Let $k_1 \in \C_-$ be such that the following inequalities hold true:
    \[ \label{p6e10}
        |k_1| \geq |k_o|,\qq \Re k_1^2 \leq \Re k_o^2.
    \]
    Then there exists $q \in \cP$ such that $a(\cdot,q) \in \cA_{symm}$ and
    $$
        a(k,q) = \frac{(k - k_1)(k + \ol k_1)}{(k - k_o)(k + \ol k_o)} a(k,q^o),\qq k \in \C.
    $$
\end{theorem}
\begin{remark}
    Let $k_o = u_o + i v_o$ and $k_1 = u_1 + i v_1$. Then inequalities (\ref{p6e10})
    have the following form (see also Fig.~\ref{fig1}):
    \[ \label{p6e12}
        u_1^2 + v_1^2 \geq u_o^2 + v_o^2,\qq u_1^2 - v_1^2 \leq u_o^2 - v_o^2.
    \]
    Note that in any neighborhood of $k_o$ there exist $k_1$
    such (\ref{p6e12}) holds true and $k_2$ such that (\ref{p6e12}) does not hold true.
\end{remark}
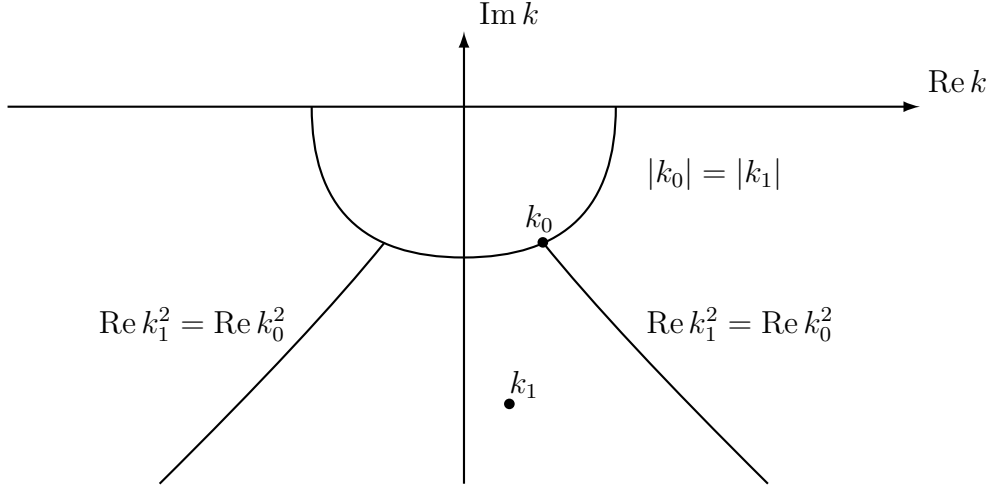
\begin{figure}[h!]
    \begin{center}
        \setlength{\unitlength}{2cm}
        \begin{picture}(6,3)
            \thicklines
            \put(0,2.5){\vector(1,0){6}} % x-axis
            \put(3,0){\vector(0,1){3}} % y-axis

            \qbezier(2,2.5)(2,1.5)(3,1.5) % left quater circle
            \qbezier(3,1.5)(4,1.5)(4,2.5) % right quater circle

            \put(3.52,1.6){\circle*{0.07}} % first resonance

            \qbezier(3.52,1.6)(4,1)(5,0) %right hyperbola

            \qbezier(2.48,1.6)(2,1)(1,0) %left hyperbola

            \put(3.4,1.7){$k_0$}

            \put(4.2,2){$|k_0| = |k_1|$}

            \put(4.2,1){$\Re k_1^2 = \Re k_0^2$}

            \put(0.6,1){$\Re k_1^2 = \Re k_0^2$}

            \put(6.05, 2.6){$\Re k$}

            \put(3.1,3.05){$\Im k$}

            \put(3.3,0.6){$k_1$}
            \put(3.3,0.53){\circle*{0.07}}
        \end{picture}
        \label{fig1}
        \caption{Example of $k_1$ and $k_0$, which satisfy (\ref{p6e10}).}
    \end{center}
\end{figure}

\section{Preliminary} \label{p2}

\subsection{Notations}
In this section, we recall results about inverse scattering problem
for the Dirac operator on the real line. We introduce the following
Banach spaces
$$
    \begin{aligned}
        \cL_{\pm} &= L^2(\R_{\pm}) \cap L^1(\R_{\pm}),\qq
        \| \cdot \|_{\cL_{\pm}} = \| \cdot \|_{L^2(\R_{\pm})} + \| \cdot \|_{L^1(\R_{\pm})},\\
        \cL &= L^2(\R) \cap L^1(\R),\qq
        \| \cdot \|_{\cL} = \| \cdot \|_{L^2(\R)} + \| \cdot \|_{L^1(\R)}.
    \end{aligned}
$$
We also define the following Banach algebras with pointwise multiplication
$$
    \begin{aligned}
        \hat\cL_{\pm} &= \{ \,\cF g \, \mid \, g \in \cL_{\pm} \,\},\qq
        \| \cF g \|_{{\hat \cL}_{\pm}} = \| g \|_{\cL_{\pm}},\\
        \hat \cL &= \{ \,\cF g \, \mid \, g \in \cL \,\},\qq
        \| \cF g \|_{{\hat \cL}} = \| g \|_{\cL},\\
        {\cW}_{\pm} &= \{ \, c + g \, \mid \, (c,g) \in \C \ts \hat \cL_{\pm} \,\},\qq
        \| c + g \|_{{\cW}_{\pm}} = |c| + \| g \|_{\hat \cL_{\pm}},\\
        {\cW} &= \{ \, c + g \, \mid \, (c,g) \in \C \ts \hat \cL \,\},\qq
        \| c + g \|_{{\cW}} = |c| + \| g \|_{\hat \cL},
    \end{aligned}
$$
It is well-known that ${\cW}_+$ and ${\cW}$ are unital Banach
algebras (see e.g. Chapter 17 in \cite{GRS64}). We denote by
$\cM_2(\C)$ the space of $2 \times 2$ matrices with
complex entries.

\subsection{Jost solutions}

We consider Dirac operator $H y = -i \s_3 y' + i \s_3 Q y$ on $L^2(\R,\C^2)$,
where the potential $Q$ has form (\ref{intro:potential}) and $q \in \cL$.
For $q \in \cL$ and $k \in \R$, we also introduce the Jost solutions $f^{\pm}(x,k)$
of Dirac equation
\[ \label{p3e2}
    (f^{\pm})'(x,k) = Q(x) f^{\pm}(x,k) + i k \s_3 f^{\pm}(x,k),\qq x \in \R,
\]
satisfying asymptotic conditions:
\[ \label{p3e1}
    f^{\pm}(x,k) = e^{i k x \s_3} \left( 1 + o(1) \right)\qq \text{as $x \to \pm \iy$}.
\]
For each $q \in \cL$ and $k \in \R$ there exist Jost solitons and they have
integral representations in terms of the transformation operators.
We recall these known results, see p.39 in  \cite{FT07} and  Proposition 3.5 in \cite{FHMP09}.

\begin{lemma} \label{hll1}
    Let $q \in \cL$. Then there exist functions $\G^{\pm}:\R \ts \R_{\pm} \to \cM_2(\C)$ such that
    \[ \label{p3e16}
        \G^{\pm} = \ma \G^{\pm}_{11} & \G^{\pm}_{12} \\ \G^{\pm}_{21} & \G^{\pm}_{22} \am,\qq
        \G^{\pm}_{11} = \ol{\G^{\pm}_{22}},\qq \G^{\pm}_{21} = \ol{\G^{\pm}_{12}},
    \]
    and
    \[ \label{hle2}
        f^{\pm}(x,k) = e^{i x k \s_3} \pm \int_0^{\pm \iy} \G^{\pm}(x,s) e^{i (2s+x) k \s_3} ds,
        \qq (x,k) \in \R^2.
    \]
    Moreover, for each $n,m = 1,2$, the following statements hold true:
    \begin{enumerate}[label={\roman*)}]
        \item For each fixed $q \in \cL$, the mappings $x \mapsto \G^{\pm}_{nm}(x, \cdot,q)$ from $\R$ into $\cL_{\pm}$
        are continuous and, for any $x \in \R$, they satisfy:
        \[ \label{hle3}
            \| \G^{\pm}_{nm}(x,\cdot,q) \|_{\cL_{\pm}} \leq e^{\eta^{\pm}(x)}(1 + \nu^{\pm}(x)) - 1,
        \]
        where
        $$
            \eta^{\pm}(x) = \left| \int_x^{\pm \iy} |q(s)| ds \right|,\qq
            \nu^{\pm}(x) = \left| \int_x^{ \pm \iy} |q(s)|^2 ds \right|^{1/2};
        $$
        \item For each fixed $x \in \R$, the mappings
        $q \mapsto \G^{\pm}_{nm}(x, \cdot, q)$ from $\cL$ into
        $\cL_{\pm}$ are continuous;
        \item For each fixed $q \in \cL$, the mappings
        $s \mapsto \G^{\pm}_{12}(\cdot, s, q)$ from $\R_{\pm}$ into $\cL$ are continuous
        and satisfy:
        \[ \label{p3e15}
            q(x) = \G^{-}_{12}(x,0,q) = - \G^{+}_{12}(x,0,q),\qq x \in \R.
        \]
    \end{enumerate}
\end{lemma}
This lemma gives that there exist matrix-valued Jost solutions only for $k \in \R$. However,
we can construct the vector-value Jost solutions for any $k \in \C_+$ or $k \in \C_-$.
It follows from the Paley-Wiener Theorem and representation (\ref{hle2}). Thus, we give the
following lemma (see e.g. Proposition 3.7 in \cite{FHMP09}).
\begin{lemma} \label{p3l8}
    Let $q \in \cL$. Then, for any fixed $x \in \R$, the functions $f^{\pm}_{11}(x,\cdot,q)$ and
    $f^{\pm}_{21}(x,\cdot,q)$ admit analytical continuation from $\R$ onto $\C_{\pm}$ and
    the functions $f^{\pm}_{12}(x,\cdot,q)$ and
    $f^{\pm}_{22}(x,\cdot,q)$ admit analytical continuation from $\R$ onto $\C_{\mp}$.
\end{lemma}

\subsection{Transition matrix}
Since equation (\ref{p3e2}) has exactly one linear independent matrix-valued solution, it follows
that for any $k \in \R$ there exists a unique $2 \times 2$ \textit{transition matrix}
$A(k) \in \cM_2(\C)$ such that
\[ \label{p3e17}
    f^{+}(x,k) = f^{-}(x,k) A(k),\qq x \in \R.
\]
Using Lemma \ref{hll1}, we obtain known representation of the matrix $A(k)$. In order to formulate
this result, we introduce the following classes.

\begin{definition*}
    $\cB_{\bu} = \hat \cL$ is a metric space equipped with the metric
    $$
        \rho_{\cB_{\bu}}(b_1,b_2) = \|\hat b_1 - \hat b_2\|_{\cL},\qq b_1,b_2 \in \cB_{\bu}.
    $$
\end{definition*}

\begin{definition*}
    $\cA_{\bu}$ is a metric space equipped with the metric $\rho_{\cA_{\bu}}$,
    where $\cA_{\bu}$ is the set of all analytic on $\C_+$ and continuous on $\ol{\C_+}$
    functions $a$ such that:
    \begin{enumerate}[label={\roman*)}]
        \item $a(k) \neq 0$ for any $k \in \C_+$;
        \item $|a(k)| \geq 1$ for any $k \in \R$;
        \item $|a|^2 - 1 \in L^1(\R)$;
        \item $a - 1 \in \hat \cL_+$;
    \end{enumerate}
    and the metric $\rho_{\cA_{\bu}}$ is given by
    $$
        \rho_{\cA_{\bu}}(a_1,a_2) = \|\cF^{-1}(a_1-a_2)\|_{\cL_+},\qq a_1,a_2 \in \cA_{\bu}.
    $$
\end{definition*}
\begin{remark}
    Note that $\cB_{\bu} = \hat \cL \ss \cW$ and $\cA_{\bu} \ss \cW_+ \ss \cW$ isometrically.
\end{remark}

Thus, we formulate the following known lemma (see e.g. \cite{FHMP09}).

\begin{lemma} \label{hll4}
    Let $q \in \cL$ and let $A$ be given by (\ref{p3e17}). Then $A$ has the following form:
    \[ \label{p3e25}
        A(k) = \ma a & \ol b \\ b & \ol a \am(k),\qq |a(k)|^2 - |b(k)|^2 = 1,\qq k \in \R,\qq (a,b) \in \cA_{\bu} \ts \cB_{\bu}
    \]
    and the mappings $q \mapsto a(\cdot,q)$ from $\cL$ into $\cA_{\bu}$ and $q \mapsto b(\cdot,q)$ from
    $\cL$ into $\cB_{\bu}$ are continuous.

    Moreover, let $a = 1 + \cF h$ for some $h \in \cL_+$.
    Then the following representations hold true:
    \[ \label{p3e23}
        \begin{aligned}
            \hat b(s) &= \G^{-}_{22}(0,-s) + \G^{+}_{11}(0,s) + \G^{-}_{22}(0,-s) * \G^{+}_{11}(0,s) -
            \G^{-}_{12}(0,-s) * \G^{+}_{21}(0,s),\, s \in \R_+, \\
            h(s) &= \G^{+}_{21}(0,s) - \G^{-}_{21}(0,s) + \G^{-}_{11}(0,s) * \G^{+}_{21}(0,s) -
            \G^{-}_{21}(0,s) * \G^{+}_{11}(0,s),\, s \in \R.
        \end{aligned}
    \]
\end{lemma}
\begin{proof}
    Substituting (\ref{hle2}) in (\ref{p3e17}) and using (\ref{p3e16}), we get that
    the matrix $A(k)$ have the following form:
    $$
        A(k) = \ma a & \ol b \\ b & \ol a \am(k),\qq k \in \R,
    $$
    where $a = 1 + \cF h$ and $b = \cF \hat b$ and representation (\ref{p3e23}) holds true.
    By Lemma \ref{hll1}, i), we have $\G^{\pm}_{nm}(0,\cdot) \in \cL_{\pm}$ for each $n,m = 1,2$
    and then $\hat b \in \cL$, $h \in \cL_+$. Thus, we get $b \in \cB_{\bu}$.

    Now, we check the other conditions of $\cA_{\bu}$.
    Since $\det f^{\pm} = 1$, it follows that $\det A = 1$, which yields $|a|^2 - |b|^2 = 1$.
    This implies that $|a(k)| \geq 1$ for any $k \in \R$ and
    $|a|^2-1 \in L^1(\R)$. Using (\ref{p3e17}), we obtain the following representation:
    \[ \label{p3e26}
        a(k) = \det \ma f^+_{11} & f^-_{12} \\ f^+_{21} & f^-_{22} \am (x,k),\qq (x,k) \in \R^2.
    \]
    Due to Lemma \ref{p3l8} and (\ref{p3e26}), the function $a$ admit an analytical continuation
    from $\R$ onto $\C_+$. We show that $a(k) \neq 0$ for any $k \in \C_+$.
    It follows from (\ref{p3e26}) that if $a(k) = 0$ for some $k \in \C_+$,
    then the vector-valued solutions
    $\left( \begin{smallmatrix} f^+_{11} \\ f^+_{21} \end{smallmatrix} \right)(x,k)$ and
    $\left( \begin{smallmatrix} f^-_{12} \\ f^-_{22} \end{smallmatrix} \right)(x,k)$
    are linearly dependent. Moreover, due to (\ref{p3e1}), they exponentially decrease as
    $x \to \pm \iy$ for any $k \in \C_+$. Thus,
    if $a(k) = 0$ for some $k \in \C_+$, then $k$ is an eigenvalue of $H$. Since $H$ is
    self-adjoint, it has no eigenvalues in $\C_+$. Thus, $a(k) \neq 0$ for any $k \in \C_+$ and
    then $a \in \cA_{\bu}$.

    Finally, it follows from Lemma \ref{hll1}, ii), that the mappings
    $q \mapsto \G^{\pm}_{nm}(x, \cdot, q)$ from $\cL$ into
    $\cL_{\pm}$ are continuous for each $n,m = 1,2$ and each fixed $x \in \R$. Due to (\ref{p3e23}),
    we get that the mappings $q \mapsto a(\cdot,q)$ from $\cL$ into $\cA_{\bu}$ and
    $q \mapsto b(\cdot,q)$ from $\cL$ into $\cB_{\bu}$ are continuous.
\end{proof}

Note that $f \in \cW_+$ is invertible if and only if $f(k) \neq 0$ for any
$k \in \ol \C_+$ (see e.g. Lemma 2.9 in \cite{HM16}).
Moreover, the inverse mapping is continuous on
the subspace of invertible elements of Banach algebra (see e.g. Chapter 2 in \cite{GRS64}).
Since $a \in \cA_{\bu}$, we get the following.

\begin{corollary} \label{hll2}
    Let $q \in \cL$ and let $a = a(\cdot,q)$. Then the mapping
    $a \mapsto a^{-1}$ from $\cA_{\bu}$ into $\cW_+$ is continuous.
    In particular, there exists a unique $g \in \cL_+$ such that
    \[ \label{p3e10}
        a^{-1}(k) = 1 + \int_0^{+\iy} g(s) e^{2 i k s} ds,\qq k \in \ol \C_+.
    \]
\end{corollary}

Now, we show that $a$ and $b$ are not independent.
In order to get this result, we need the following known lemma about the Banach algebra $\cW$
(see e.g. Chapter 6 in \cite{GRS64}).
\begin{lemma} \label{a1l1}
    Let $f \in \cW$ and let $\Omega$ be an open neighborhood of the closure of the range
    of $f$. Let $\phi$ be an analytic function on $\Omega$. Then we have $\phi \circ f \in \cW$.
    Moreover, the mapping $f \mapsto \phi \circ f$ is an continuous mapping on the subspace
    of all function $f \in \cW$ such that its range contained in $\Omega$.
\end{lemma}
In particular, we obtain the following corollary for the logarithm and the exponential function.
We introduce subspaces of $\cW$:
$$
    \begin{aligned}
        \cW_{real} &= \{\, f = 1 + g \, \mid \, g \in \hat \cL,\, f(x) > 0,\, x \in \R \, \},\\
        \cW_1 &= \{\, f = 1 + g \, \mid \, g \in \hat \cL,\, |f(x)| > 0,\, x \in \R \, \}.
    \end{aligned}
$$
We also introduce the mappings $\exp: f \mapsto e^{f(\cdot)}$, $f \in \hat \cL$, and
$\log : f \mapsto \log(f(\cdot))$, $f \in \cW_{real}$, where we
fixed the branch of the logarithm by $\log(x) \in \R$ for any $x \in \R$.
\begin{corollary} \label{p3c1}
    The mappings $\log : \cW_{real} \to \hat \cL$ and $\exp : \hat \cL \to \cW_1$
    are continuous.
\end{corollary}
\begin{proof}
    For any $f \in \cW_{real}$, the range of $f$ is a compact subset of $\R_+$ and then,
    by Lemma \ref{a1l1}, $\log$ is a continuous mapping from $\cW_{real}$ to $\cW$.
    Moreover, it follows from the Riemann-Lebesgue lemma (see e.g. Theorem IX.7 in \cite{RS80})
    that $f(x) \to 1$ as $x \to \pm \iy$ for any $f \in \cW_{real}$ and
    then $\log(f(x)) \to 0$ as $x \to \pm \iy$. Thus, we have $\log(f) \in \hat \cL$.

    For any $f \in \hat \cL$, the range of $f$ is a compact subset of $\C$ and then, by Lemma \ref{a1l1},
    $\exp$ is a continuous mapping from $\hat \cL$ to $\cW$. As above, $f(x) \to 0$ as
    $x \to \pm \iy$ for any $f \in \hat \cL$ and then $\exp(f(x)) \to 1$ as $x \to \pm \iy$.
    Since $\exp(x) \neq 0$ for any $x \in \C$, it follows that $\exp(f) \in \cW_1$.
\end{proof}
We also introduce the Cauchy integral operator $\cC_+$ on $L^2(\R)$ by
$$
    (\cC_+ f)(k) = \lim_{\ve \to +0} \frac{1}{2\pi i} \int_{\R} \frac{f(s)}{s - k - i \ve} ds,
        \qq k \in \R,
$$
where the limit is taken in $L^2(\R)$. It is well-known that $\cC_+$ has the representation:
\[ \label{p3e3}
    \cC_+ = \cF \chi_+ \cF^{-1},
\]
where $\chi_+$ is the indicator function of $\R_+$. We need the following known lemma.
\begin{lemma} \label{a1l2}
    The mapping $f \mapsto \cC_+ f$ from $\hat \cL$ into $\hat \cL_+$ is continuous.
\end{lemma}
\begin{proof}
    Let $f = \cF g \in \hat \cL$.
    Using (\ref{p3e3}), we get $\cC_+ f = \cF (\chi_+ g) \in \hat \cL_+$. Since
    the mapping $g \mapsto \chi_+ g$ from $\cL$ to $\cL_+$ is continuous, we get the statement
    of the lemma.
\end{proof}
We also need the following technical lemma.
\begin{lemma} \label{p3l1}
    Let $h \in \hat \cL \cap L^1(\R)$ and let $h(k) \geq 0$ for any $k \in \R$.
    Then there exists a unique solution $a \in \cA_{\bu}$ of the equation $|a|^2 = 1+h$.
    Moreover, the mapping $h \mapsto a$ from $\hat \cL$ into $\cA_{\bu}$ is continuous.
\end{lemma}
\begin{proof}
    Firstly, we prove the existence and continuity.
    It is easy to see that $1+h \in \cW_{real}$. Thus,
    it follows from Corollary \ref{p3c1} that
    $\frac{1}{2} \log(1+h) \in \hat \cL \ss L^2(\R)$ and the mapping
    $h \mapsto \frac{1}{2} \log(1+h)$ from $\hat \cL$ into $\hat \cL$ is continuous.
    We introduce
    $$
        F(k) = \frac{i}{2\pi} \int_{\R} \frac{\log(1+h)}{k - t}dt,\qq k \in \C_+.
    $$
    It follows from Corollary on p.~128 in~\cite{Koo2} that $F$ is analytic on $\C_+$
    and $\Re F(k) \to \frac{1}{2} \log(1+h(t))$ as $k \to t$ and $k \in \C_+$, for almost all $t \in \R$.
    Moreover, it follows from Lemma \ref{a1l2} that $F \in \hat \cL_+$, where
    $F(t) = \lim_{k \to t} F(k)$, $t \in \R$ and $F$ depends continuously on $h \in \hat \cL$.

    Now, let $a = \exp(F)$. Then it follows from Corollary \ref{p3c1} that $a \in \cW_1$ and $a$
    depends continuously on $F \in \hat \cL_+ \ss \hat \cL$. It also follows that $a$ is analytic in $\C_+$ and
    $a(k) \neq 0$ for any $k \in \C_+$. Since $h(k) \geq 0$ for any $k \in \R$ and $h \in L^1(\R)$,
    we have that $a(k) \geq 1$ for any $k \in \R$ and $|a|^2-1 \in L^1(\R)$. Thus, we have $a \in \cA_{\bu}$.
    Recall that $F$ depends continuously on $h$ and then $a$ also depends continuously on $h$.

    Secondly, we prove the uniqueness. Let $a \in \cA_{\bu}$ be a solution of the equation $|a|^2 = 1+h$.
    Then $\log(a)$ is an analytic function in $\C_+$ and then $\Re \log a$ is a harmonic function in
    $\C_+$ and it is a uniquely determined by its boundary value $\frac{1}{2} \Re \log (1+h) \in L^2(\R)$.
    Since $\C_+$ is a connected open subset, it follows that $\Re \log a$ has a unique harmonic conjugate up to
    the constant. Thus, if $a_1$ is an other solution of the equation $|a|^2 = 1+h$,
    which is analytic in $\C_+$ and $a_1(k) \to 1$ as $k \to \pm \iy$,
    then we have $\log a = \log a_1 + iC$ for some $C \in \R$. Due to $a(k) \to 1$ and $a_1(k) \to 1$
    as $k \to \pm \iy$, it follows that $\log a(k) \to 2\pi i n$ and $\log a_1(k) \to 2\pi i n_1$ as $k \to \pm \iy$
    for some $n,n_1 \in \Z$. Then $\log a(k) = \log a_1(k) + 2 \pi i m$ for any
    $k \in \ol \C_+$ for some $m \in \Z$ and then $a = a_1$.
\end{proof}
Now, we show that for any $b \in \cB_{\bu}$ there exists a unique $a \in \cA_{\bu}$ such that
the determinant of the associated transition matrix equals $1$.
\begin{lemma} \label{p3l6}
    Let $b \in \cB_{\bu}$. Then there exists a unique solution $a \in \cA_{\bu}$ of the equation
    \[ \label{p3e5}
        |a(k)|^2 - |b(k)|^2 = 1,\qq k \in \R.
    \]
    Moreover, the mapping $b \mapsto a$ from $\cB_{\bu}$ in $\cA_{\bu}$ is continuous.
\end{lemma}
\begin{proof}
    Let $h = |b|^2$. Then it follows from the basic properties of Banach algebras that
    $h \in \hat \cL$ and $h$ depends continuously on $b \in \cB_{\bu} = \hat \cL$.
    Due to $b \in \hat \cL \ss L^2(\R)$, we have $h = |b|^2 \in L^1(\R)$. At last, we have
    $h(k) \geq 0$ for any $k \in \R$. Thus, it follows from Lemma \ref{p3l1} that there exists
    a unique solution $a \in \cA_{\bu}$ of the equation (\ref{p3e5}) and it depends continuously on $b$.
\end{proof}

\begin{remark}
    Note that the similar result for compactly supported potentials
    can be obtained using the theory of entire functions (see e.g. Theorem 2.3 in \cite{K05}).
\end{remark}

\subsection{Direct scattering}

Recall that $H_o y = -i \s_3 y'$ on $L^2(\R,\C^2)$ is the free Dirac operator.
The scattering matrix for the pair $H_o$, $H$ has the following form
$$
    S(z) = \frac{1}{a(z)} \ma 1 & -\ol b(z) \\ b(z) & 1\am,\qq z \in \R.
$$
Here $1/a$ is the transmission coefficient and $r_+ = -\ol b/a$ (or $r_- = b/a$) is the right
(or left) reflection coefficient. We introduce the following class of all reflection coefficients.
\begin{definition*}
    $\cR_{\bu}$ is a metric space of all functions $r \in \hat \cL$ such that
    $|r(k)| < 1$ for any $k \in \R$ equipped with the metric
    $$
        \rho_{\cR_{\bu}}(r_1,r_2) = \|\hat r_1 - \hat r_2\|_{\cL},\qq r_1,r_2 \in \cR_{\bu}.
    $$
\end{definition*}

It follows from the definition of the scattering matrix that it can be obtained from
the transition matrix. Due to Corollary \ref{hll2}, the transmission coefficient
$\frac{1}{a} \in \cW_+$ and it depends continuously on $a \in \cA_{\bu}$. Now, we show that the reflection
coefficient $r_{+}$ also depends continuously on $a \in \cA_{\bu}$ and $b \in \cB_{\bu}$. Note that this
result is also well-known (see e.g. \cite{FHMP09}).
\begin{lemma} \label{p3l3}
    Let $a \in \cA_{\bu}$, $b \in \cB_{\bu}$ be such that $|a|^2 - |b|^2 = 1$.
    Then we have $r_+ = -\ol b/a \in \cR_{\bu}$ and $r_+$ depends continuously on $a \in \cA_{\bu}$ and
    $b \in \cB_{\bu}$.
\end{lemma}
\begin{proof}
    Since $|a|^2 =1 + |b|^2$ and $|r_+| = \frac{|b|}{|a|}$, we have
    $$
        |r_+|^2 \leq \frac{|b|^2}{1 + |b|^2} < 1.
    $$
    By Lemma \ref{hll2}, $a^{-1} \in \cW_+$. Due to $b \in \cB_{\bu} = \hat \cL$, we get
    $\ol b \in \hat \cL$ and then $r_+ \in \hat \cL$. It follows from Corollary \ref{hll2} that
    $a \mapsto a^{-1}$ is a continuous mapping from $\cA_{\bu}$ to $\cW_+$.
    Moreover, the multiplication and the complex conjugate are also continuous mappings
    from $\hat \cL$ to $\hat \cL$. Then $r_+$ depends continuously on $a$ and $b$.
\end{proof}
On the other hand, the coefficients $a$ and $b$ are uniquely determined by the reflection coefficient.
\begin{lemma} \label{p3l4}
    Let $r_+ \in \cR_{\bu}$. Then there exist a unique solution $a \in \cA_{\bu}$ of the equation
    \[ \label{p3e6}
        1 - |r_+|^2 = \frac{1}{|a|^2}
    \]
    and the mapping $r_+ \mapsto a$ from $\cR_{\bu}$ into $\cA_{\bu}$ is continuous. Moreover, in this case
    $b = - \ol{r_+} \ol{a} \in \cB_{\bu}$ and $b$ depends continuously on $r_+$.
\end{lemma}
\begin{proof}
    It follows from (\ref{p3e6}) that
    $$
        |a|^2 = \frac{1}{1 - |r_+|^2} = 1 + \frac{|r_+|^2}{1 - |r_+|^2} = 1 + h.
    $$
    Due to $|r_+(k)| < 1$, for any $k \in \R$, we get $h(k) \geq 0$ for any $k \in \R$.
    Using the multiplication properties of the Banach algebra, we have $|r_+|^2 \in \hat \cL$.
    Note that $f \in \cW$ is invertible if and only if $f(k) \neq 0$ for any
    $k \in \R$ (see e.g. Lemma 2.9 in \cite{HM16}). Since $1 - |r_+|^2 \in \cW$ and
    $1 - |r_+(k)|^2 > 0$ for any $k \in \R$, it follows that
    $(1 - |r_+|^2)^{-1} \in \cW$ and then $h \in \hat \cL$ and it depends continuously on $r_+ \in \cR_{\bu}$.
    It follows from the Riemann-Lebesgue lemma that $r_+(k) \to 0$ as $k \to \pm \iy$.
    Since $r_+ = \cF g$ for some $g \in \cL$, we have that $r_+ \in C(\R)$. Recall that
    $1 - |r_+(k)|^2 > 0$ for any $k \in \R$. Hence, we obtain $(1 - |r_+|^2)^{-1} \in L^{\iy}(\R)$.
    Due to $r_+ \in L^2(\R)$, we have $|r_+|^2 \in L^1(\R)$ and then $h \in L^1(\R)$.
    Thus, by Lemma \ref{p3l1}, there exists a unique solution $a \in \cA_{\bu}$ of the equation (\ref{p3e6})
    and it depends continuously on $r_+ \in \cR_{\bu}$.
    Finally, using the multiplication properties of the Banach algebra, we get
    $b = - \ol{r_+} \ol{a} \in \cB_{\bu}$ and $b$ depends continuously on $r_+$.
\end{proof}

Above, we considered only the right reflection coefficient. However,
the left reflection coefficient is uniquely determined by the right one.
We introduce the mapping $\cI: r_+ \mapsto r_- = -\ol r_+ \frac{\ol a}{a}$ from $\cR_{\bu}$ in $\cR_{\bu}$,
where $a$ is given by Lemma \ref{p3l4}. The following result holds true (see e.g. Lemma 3.4 in \cite{FHMP09}).

\begin{lemma} \label{p3l5}
    The mapping $\cI$ is a homeomorphism and $\cI \circ \cI = I_{\cR_{\bu}}$, where $I_{\cR_{\bu}}$ is
    the identity mapping on $\cR_{\bu}$.
\end{lemma}
\begin{remark}
    It follows from Lemmas \ref{p3l4}, \ref{p3l5} that we can recover the scattering matrix from
    one reflection coefficient. Moreover, by Lemmas \ref{p3l6}, \ref{p3l3}, we can recover
    the reflection coefficient from the coefficient $b \in \cB_{\bu}$ and then the scattering matrix too.
\end{remark}

\subsection{Inverse scattering}
Now, we give the solution of the inverse scattering problem for the Dirac operators
(see e.g. Theorem 1.1 in \cite{FHMP09}).
\begin{theorem} \label{hlt1}
    The mappings $q \mapsto r_{\pm}(\cdot,q)$ are homeomorphisms between $\cL$ and $\cR_{\bu}$.
\end{theorem}
It follows from this theorem that a potential is uniquely determined
by the right or left reflection coefficient. Thus, it solves the uniqueness, the characterization
and the continuity problems for potentials from $\cL$ in terms of the reflection coefficients.
Using this result, we solve the inverse problem in terms of the coefficient $b$.
\begin{theorem} \label{hlt2}
    The mapping $q \mapsto b(\cdot,q)$ is a homeomorphism between $\cL$ and $\cB_{\bu}$.
\end{theorem}
\begin{proof}
    We consider the composition of mappings:
    $$
        q \mapsto r_+(\cdot,q) \mapsto \ma r_+(\cdot,q) \\ a(\cdot, q) \am \mapsto b(\cdot,q).
    $$
    Using Theorem \ref{hlt1} and Lemma \ref{p3l4}, we get that for any $q \in \cL$ there exists
    a unique $b(\cdot,q) \in \cB_{\bu}$ and $b(\cdot,q)$ depends continuously on $q$.
    On the other side, we consider the composition of mappings:
    $$
        b \mapsto \ma b \\ a \am \mapsto r_+ \mapsto q.
    $$
    Using Lemmas \ref{p3l6},\ref{p3l3}, and Theorem \ref{hlt1}, we get that for any $b \in \cB_{\bu}$
    there exists a unique $q \in \cL$ such that $b = b(\cdot,q)$ and $q$ depends continuously on $b$.
\end{proof}

In order to recover a potential from the reflection coefficient, one can use
the Gelfand-Levitan-Marchenko (GLM) equation.
For any $r_{\pm} \in \cR_{\bu}$, we introduce the matrix-valued functions
\[ \label{hlF_scat}
    \Omega_{+}(s) = \ma 0 & F_{+}(-s) \\ \ol{F_{+}(-s)} & 0\am,\qq
    \Omega_{-}(s) = \ma 0 & \ol{F_{-}(s)} \\ F_{-}(s) & 0\am,\qq s \in \R,
\]
where $F_{\pm} = \hat r_{\pm} \in \cL$. The following result was also obtained in \cite{HM16}.
\begin{lemma} \label{hll3}
    \begin{enumerate}[label={\roman*)}]
        \item Let $\G^{\pm}(x,s) = \G^{\pm}(x,s,q)$ and $\Omega_{\pm}(s) = \Omega_{\pm}(s,q)$ for
        some $q \in \cL$ and for any $(x,s) \in \R \ts \R_{\pm}$. Then $\G^{\pm}$ and $\Omega_{\pm}$
        satisfy the GLM equations:
        \begin{align}
            \G^+(x,s) + \Omega_+(x+s) + \int_0^{+\iy} \G^+(x,t) \Omega_+(x+t+s) dt &= 0, \label{GLM+} \\
            \G^-(x,s) + \Omega_-(x+s) + \int_{-\iy}^{0} \G^-(x,t) \Omega_-(x+t+s) dt &= 0 \label{GLM-}
        \end{align}
        for each $x \in \R$ and almost all $s \in \R_{\pm}$.
        \item Let $\Omega_{+}$ be given by (\ref{hlF_scat}) for some $r_+ \in \cR_{\bu}$.
        Then equation (\ref{GLM+}) has a unique solution $\G^+(x,\cdot) \in \cL_+ \otimes \cM_2(\C)$
        and this solution depends continuously on $x \in \R$.
        Moreover, the mapping $s \mapsto \G^+_{12}(\cdot,s)$ from $\R_+$ into $\cL$ is continuous.
        \item Let $\Omega_{-}$ be given by (\ref{hlF_scat}) for some $r_- \in \cR_{\bu}$.
        Then equation (\ref{GLM-}) has a unique solution $\G^-(x,\cdot) \in \cL_- \otimes \cM_2(\C)$
        and this solution depends continuously on $x \in \R$.
        Moreover, the mapping $s \mapsto \G^-_{12}(\cdot,s)$ from $\R_-$ into $\cL$ is continuous.
        \item Let $r_+, r_- \in \cR_{\bu}$ such that $r_+ = \cI r_-$. Then we get
        $$
            \G^-_{12}(x,0) = -\G^+_{12}(x,0),\qq x \in \R.
        $$
    \end{enumerate}
\end{lemma}
\begin{remark}
    One can recover a potential $q$ from the reflection coefficient $r_+ \in \cR_{\bu}$ as follows:
    \begin{enumerate}[label={(\roman*)}]
        \item Construct $\Omega_+$ by $r_+$ as in (\ref{hlF_scat});
        \item Construct $\G^+(x,s)$ as a solution of (\ref{GLM+});
        \item Recover a potential by using $q(x) = -\G^+_{12}(x,0)$, $x \in \R$.
    \end{enumerate}
\end{remark}

\section{Compactly supported potentials} \label{p3}
In this section, we show the relationship between the support of a potential and properties
of the transition and scattering matrix. Firstly, we show that a potential is compactly supported
if and only if the associated kernels $\G^{\pm}$ are compactly supported.

\begin{lemma} \label{p4l1}
    Let $q \in \cL$ and let $\G^{\pm}(x,s) = \G^{\pm}(x,s,q)$. Then we have for any $\d > 0$:
    \begin{enumerate}[label={\roman*)}]
        \item $\sup \supp q \leq \d$ if and only if $\G^{+}(x,s) = 0$
        for almost all $(x, s) \in \R^2_+$ such that $x + s > \d$;
        \item $\inf \supp q \geq -\d$ if and only if $\G^{-}(x,s) = 0$
        for almost all $(x, s) \in \R^2_-$ such that $x + s < -\d$.
    \end{enumerate}
\end{lemma}
\begin{proof}
    i) Let $\sup \supp q \leq \d$. Then it follows from (\ref{hle3}) that $\| \G^{+}_{nm}(x,\cdot) \|_{\cL_+} = 0$
    for each $x > \d$, $n,m = 1,2$. Thus, $\G^{+}(x,s) = 0$ for each $x > \d$ and for almost all $s \in \R_+$.
    Substituting this identity in (\ref{GLM+}), we get $\Omega_{+}(x+s) = 0$ for each $x > \d$ and
    for almost all $s \in \R_+$, i.e. $\Omega_{+}(x) = 0$ for almost all $x > \d$.
    Now substituting this identity in (\ref{GLM+}), we get $\G^{+}(x,s) = 0$ for almost all
    $x,s \in \R_+$ such that $x+s > \d$.

    Let $\G^{+}(x,s) = 0$ for almost all $x,s \in \R_+$ such that $x+s > \d$. By Lemma \ref{hll3}
    the mapping $s \mapsto \G^{+}_{12}(\cdot,s)$ is continuous. Combining these facts, we get
    $\G^{+}_{12}(x,0) = 0$ for almost all $x > \d$, which yields, by Lemma
    \ref{hll1}, that $q(x) = 0$ for almost all $x > \d$.

    ii) In this case, the proof is similar.
\end{proof}

The support of a potential is also related to the support of $\hat b$ and $\hat r_{\pm}$.
\begin{lemma} \label{p3l2}
    Let $q \in \cL$ and let $b = b(\cdot,q)$, $r_{\pm} = r_{\pm}(\cdot,q)$. Then we have
    \[ \label{p3e11}
        \begin{aligned}
            \sup \supp q &= -\inf \supp \hat r_+ = \sup \supp \hat b,\\
            \inf \supp q &= \inf \supp \hat r_- = \inf \supp \hat b.
        \end{aligned}
    \]
\end{lemma}
\begin{proof}
    Firstly, we show that $\sup \supp q \leq -\inf \supp \hat r_+$.
    If $\inf \supp \hat r_+ = -\iy$, then the inequality is evident. Let $\inf \supp \hat r_+ = -\d < 0$
    for some $\d > +\iy$. Due to (\ref{hlF_scat}), we have $\Omega_{+}(x) = 0$ for any $x > \d$.
    Substituting this identity in (\ref{GLM+}), we get $\G^{+}(x,s) = 0$ for almost all $x,s \in \R_+$
    such that $x+s > \d$. Thus, by Lemma \ref{p4l1}, $\sup \supp q \leq \d$.

    Secondly, we show that $-\inf \supp \hat r_+ \leq \sup \supp \hat b$. Let $\sup \supp \hat b = \d$.
    It follows from (\ref{p3e10}) that
    \[ \label{p3e14}
        \hat r_+(s) = -g(s) - (g*r)(s),\qq s \in \R,
    \]
    where $g(s) = \overline{\hat b(-s)}$, $s \in \R$ and $r \in \cL_+$.
    Using $\inf \supp g = -\d$, $\inf \supp r \geq 0$, and well-known property
    $\supp (g * r) \ss \supp g + \supp r$, we get
    $\inf \supp g * r \geq -\d$. Substituting these inequalities in (\ref{p3e14}),
    we have $-\inf \supp \hat r_+ \leq \d$.

    Thirdly, we show that $\sup \supp \hat b \leq \sup \supp q$. Let $\sup \supp q = \d$. Then, by Lemma
    \ref{p4l1}, $\G^{+}(0,s) = 0$ for almost all $s > \d$. Using (\ref{p3e23}), we get $\sup \supp \hat b \leq \d$.

    Combining these three inequalities, we obtain the first line in (\ref{p3e11}). The proof of the
    second line in (\ref{p3e11}) is similar.
\end{proof}

In Lemma \ref{p3l6}, we proved that for any $b \in \cB_{\bu}$ there exists a unique $a \in \cA_{\bu}$
such that $|a|^2 - |b|^2 = 1$.
Now, we show that if $b \in \cB$, then the corresponding $a \in \cA$. Recall that $\cA$ and $\cB$ was defined
in Section \ref{p1}. Recall also that we introduced
the Cartwright classes of entire functions $\cE_{Cart}$ with fixed types $\t_{\pm}$ in Section \ref{p1}.
Now, we define the more general classes.
\begin{definition*}
    For any $\a,\b \geq 0$, $\cE_{Cart}(\a,\b)$ is a class of entire functions of exponential type $f$ such that
    $$
        \int_{\R} \frac{\log(1+|f(k)|)dk}{1 + k^2} < \iy,\qq \t_+(f) = \a,\qq \t_-(f) = \b,
    $$
    where $\t_{\pm}(f) = \lim \sup_{y \to +\iy} \frac{\log |f(\pm i y)|}{y}$.
\end{definition*}
\begin{remark}
    Note that $\cE_{Cart} = \cE_{Cart}(0,2\g)$.
\end{remark}
If $f \in \cE_{Cart}(\a,\b)$ for some $\a,\b \geq 0$, then it also has the Hadamard factorization.
Let $p \geq 0$ be the multiplicity of zero $k = 0$ of $f$.
We denote by $(k_n)_{n \geq 1}$ zeros of $f$ in $\C \sm \{ 0 \}$
counted with multiplicity and arranged that $0 < |k_1| \leq |k_2| \leq \ldots$.
Then $f$ has the Hadamard factorization
$$
    f(k) = C k^p e^{i \kappa k} \lim_{r \to +\iy}
        \prod_{|k_n| \leq r} \left(1 - \frac{k}{k_n}\right),\qq k \in \C,
$$
see, e.g., pp.127-130 in \cite{L96}, where the product converges uniformly
on compact subsets of $\C$ and
$$
    \kappa = \frac{\b - \a}{2},\qq C = \frac{f^{(p)}(0)}{p!},\qq
    \sum_{n \geq 1} \frac{|\Im k_n|}{|k_n|^2} < +\iy,\qq
    \exists \lim_{r \to +\iy} \sum_{|k_n| \leq r} \frac{1}{k_n} \neq \iy.
$$
\begin{lemma} \label{hll5}
    Let $b \in \cB$. Then there exist a unique solution $a \in \cA$ of the equation
    \[ \label{p7e4}
        a(k) a_*(k) - b(k) b_*(k) = 1,\qq k \in \C.
    \]\
    Moreover, the mapping $b \mapsto a$ from $\cB$ into $\cA$ is continuous.
\end{lemma}
\begin{proof}
    Due to Lemma \ref{p3l6}, there exists a unique solution $a_0 \in \cA_{\bu}$ of equation (\ref{p7e4})
    for $k \in \R$ and it depends continuously on $b \in \cB \ss \cB_{\bu}$.
    Thus, we only need to show that $a \in \cA$ and it is a solution of (\ref{p7e4}) for any $k \in \C$.

    Let $A(k) = 1 + b(k) b_*(k)$, $k \in \C$. Since $b \in \cB \ss \cE_{Cart}$, it follows that
    $\t_+(b) = 0$, $\t_-(b) = 2\g$ and then $A \in \cE_{Cart}(2\g,2\g)$ and the following properties hold true:
    \begin{enumerate}[label={(\roman*)}]
        \item $A_* = A$ ($A(k) = 0$ iff $A(\ol k) = 0$, $k \in \C$);
        \item $A(k) \geq 1$ for any $k \in \R$.
    \end{enumerate}
    Here we used the following simple facts about conjugate functions:
    $$
        1_* = 1,\qq b_{**} = b,\qq \t_{\pm}(b_*) = \t_{\mp}(b),\qq b(z) b_*(k) = |b(k)|^2,\qq k \in \R.
    $$
    Let $(k_n)_{n \geq 1}$, be the zeros of $A$ in $\C_-$
    counted with multiplicity and arranged
    that $0 < |k_1| \leq |k_2| \leq \ldots$. Then the Hadamard factorization for $A$ has
    the following form
    \[ \label{p7e2}
        A(k) = C \lim_{r \to +\iy}
        \prod_{|k_n| \leq r} \left(1 - \frac{k}{k_n}\right)\left(1 - \frac{k}{\ol k_n}\right),\qq k \in \C,
    \]
    where the product converges uniformly on compact subsets of $\C$, the constant $C = A(0) \geq 1$, and
    \[ \label{p7e1}
        \sum_{n \geq 1} \frac{|\Im k_n|}{|k_n|^2} < +\iy,\qq
        \lim_{r \to +\iy} \sum_{|k_n| \leq r} \frac{\Re k_n}{|k_n|^2} < +\iy.
    \]
    We introduce
    \[ \label{p7e3}
        a_o(k) = |C|^{1/2} e^{i\g k} \lim_{r \to +\iy}
        \prod_{|k_n| \leq r} \left(1 - \frac{k}{k_n}\right),\qq k \in \C.
    \]
    Due to (\ref{p7e1}), we have
    $$
        \lim_{r \to +\iy} \sum_{|k_n| \leq r} \frac{1}{k_n} < +\iy
    $$
    and then, by the Lindel{\"o}f theorem (see, e.g., p.~21~in~\cite{Koo98}), the product in
    (\ref{p7e3}) converges uniformly on compact subsets of $\C$ and
    $a_o \in \cE_{Cart}(\t_+(a_o),\t_-(a_o))$, where $\t_-(a_o) - \t_+(a_o) = 2\g$.
    Using (\ref{p7e2}) and (\ref{p7e3}), we get $a_o (a_o)_* = A$, i.e. $a_o(k)$ is a solution of
    (\ref{p7e4}) for any $k \in \C$.
    It follows from (\ref{p7e4}) that $\t_+(a_o) + \t_-(a_o) = \t_+(A) = 2\g$ and then, using
    $\t_-(a_o) - \t_+(a_o) = 2\g$, we get $\t_+(a_o) = 0$, $\t_-(a_o) = 2\g$.

    Note that in the proof of Lemma \ref{p3l1},
    we show that if $a_1$ and $a_2$ are analytic in $\C_+$, $a_j(k) \neq 0$, $k \in \C_+$, $j = 1,2$,
    and they are solutions of (\ref{p7e4}), then $a_1 = e^{i\phi}a_2$ for some $\phi \in \R$.
    Thus, it follows that $a = e^{i\phi} a_o$ for some $\phi \in \R$, which yields $a \in \cE_{Cart}$.
    Since $a \in \cA_{\bu}$, there exist $h \in \cL_+$ such that $a = 1 + \cF h$. Due to $a \in \cE_{Cart}$,
    it follows from the Paley-Wiener Theorem that $h \in \cP$ and then $a \in \cA$.
\end{proof}

Above, we recover the coefficient $a$ from $b$. Now, we show that the coefficient $b$ is uniquely
determined by $a$ with additional data.
\begin{lemma} \label{hll6}
    Let $a \in \cA$ and let $\xi \in \Xi(aa_*-1)$. Then there exists a unique solution $b \in \cB$,
    $\xi(b) = \xi$, of the equation
    \[ \label{p7e5}
        a(k) a_*(k) - b(k) b_*(k) = 1,\qq k \in \C.
    \]
\end{lemma}
\begin{proof}
    Let $B(k) = a(k) a_*(k) - 1$, $k \in \C$. Since $a \in \cE_{Cart}$, it follows that
    $\t_+(a) = 0$, $\t_-(a) = 2\g$ and then $B \in \cE_{Cart}(2\g,2\g)$ and the following properties hold true:
    \begin{enumerate}[label={(\roman*)}]
        \item $B_* = B$ ($B(k) = 0$ iff $B(\ol k) = 0$, $k \in \C$);
        \item $B(k) \geq 0$ for any $k \in \R$.
    \end{enumerate}
    Due to (ii), the real zeros of $B$ have even multiplicity.
    Let $2p \geq 0$ be the multiplicity of the zero $k = 0$ of $B$. Then it follows from (ii)
    that $B^{(2p)}(0) > 0$. Let $(z_n)_{n \geq 1}$
    be the zeros of $B$ in $\ol \C_+ \sm \{ 0 \}$ counted with multiplicity and arranged
    that $0 < |z_1| \leq |z_2| \leq \ldots$, where we take only the half of real zeros.
    Let $\xi = (\xi_n)_{n \geq 0} \in \Xi(B)$ and let
    $$
        \z_n =
        \begin{cases}
            z_n, & \xi_n = 1\\
            z_n, & \xi_n = 0\\
            \ol z_n, & \xi_n = -1
        \end{cases},\qq n \geq 1.
    $$
    Then the Hadamard factorization for $B$ has the following form
    \[ \label{p7e6}
        B(k) = C k^{2p} \lim_{r \to +\iy}
        \prod_{|\z_n| \leq r} \left(1 - \frac{k}{\z_n}\right)\left(1 - \frac{k}{\ol \z_n}\right),\qq k \in \C,
    \]
    where the product converges uniformly on compact subsets of $\C$, the constants $C > 0$, $p \geq 0$,
    and
    \[ \label{p7e7}
        \sum_{n \geq 1} \frac{|\Im \z_n|}{|\z_n|^2} < +\iy,\qq
        \lim_{r \to +\iy} \sum_{|\z_n| \leq r} \frac{\Re \z_n}{|\z_n|^2} < +\iy.
    \]
    We introduce
    \[ \label{p7e8}
        b(k) = \xi_0 |C|^{1/2} k^{p} e^{i\g k} \lim_{r \to +\iy}
        \prod_{|\z_n| \leq r} \left(1 - \frac{k}{\z_n}\right),\qq k \in \C.
    \]
    Due to (\ref{p7e7}), we have
    $$
        \lim_{r \to +\iy} \sum_{|\z_n| \leq r} \frac{1}{\z_n} < +\iy
    $$
    and then, by the Lindel{\"o}f theorem (see e.g. p.~21~\cite{Koo98}), the product in
    (\ref{p7e8}) converges uniformly on compact subsets of $\C$ and
    $b \in \cE_{Cart}(\t_+(b),\t_-(b))$, where $\t_+(b) - \t_-(b) = 2\g$.
    Using (\ref{p7e6}) and (\ref{p7e8}), we get
    $b b_* = B$, i.e. $b$ is a solution of (\ref{p7e5}). It follows from (\ref{p7e5}) that
    $\t_+(b) + \t_-(b) = \t_+(B) = 2\g$ and then $\t_+(b) = 0$, $\t_-(b) = 2\g$.
    Let $b_1$ be another solution of equation (\ref{p7e4}) such that $b_1 \in \cE_{Cart}$
    and $\xi(b_1) = \xi$. Since the function from $\cE_{Cart}$ is uniquely determined by its
    zeros and by $b^{(p)}(0)/|b^{(p)}(0)| = \xi_0$, it follows that $b_1 = b$.

    Now, we show that $b \in \cB$. Due to $a \in \cA$, we get $|b|^2 = |a|^2 - 1 \in L^1(\R)$
    and then $b \in L^2(\R)$.
    Since $\t_+(b) = 0$ and $\t_-(b) = 2\g$, it follows from the Paley-Wiener
    theorem that $b = \cF h$ for some $h \in \cP$, which yields $b \in \cB$.
\end{proof}

\section{Proof of the main theorems} \label{p4}

\begin{proof}[\bf Proof of Theorem \ref{t1}]
    In Theorem \ref{hlt2}, we have shown that the mapping $q \mapsto b(\cdot,q)$ is a homeomorphism
    between $\cL$ and $\cB_{\bu}$. Since the metrics on $\cP$ and $\cL$ and on $\cB$ and $\cB_{\bu}$ are equivalent,
    we only need to prove that the restriction of this mapping on $\cP$ is a bijection between $\cP$ and $\cB$.
    Due to Theorem \ref{hlt2}, if $q \in \cP \ss \cL$, then we have $b(\cdot,q) \in \cB_{\bu}$.
    Using Lemma \ref{p3l2}, we get
    $$
        \sup \supp \hat b = \sup \supp q = \g,\qq \inf \supp \hat b = \inf \supp q = 0,
    $$
    which yields that $\hat b \in \cP$ and then we have $b(\cdot,q) \in \cB$.

    Let $b \in \cB \ss \cB_{\bu}$. Then, by Theorem \ref{hlt2}, there exists a unique $q \in \cL$
    such that $b(\cdot,q) = b$. Since $b \in \cB$, we have $\hat b \in \cP$.
    Then it follows from Lemma \ref{p3l2} that
    $$
        \sup \supp \hat b = \sup \supp q = \g,\qq \inf \supp \hat b = \inf \supp q = 0,
    $$
    which yields that $q \in \cP$.
\end{proof}

\begin{proof}[\bf Proof of Theorem \ref{t2}]
    Let $q \in \cP \ss \cL$. By Theorem \ref{t1}, we have $b = b(\cdot,q) \in \cB$
    and then there exists a unique $\xi = \xi(b) \in \Xi(B)$, where $B = b b_*$.
    Since $b \in \cB$ and $a = a(\cdot,q)$ is a solution of $a a_* - b b_* = 1$,
    it follows from Lemma \ref{hll5} that $a(\cdot,q) \in \cA$.

    Let $(a,\xi) \in \cA \ts \Xi(B)$, where $B = a a_* - 1$. Due to Lemma \ref{hll6},
    there exists a unique solution $b \in \cB$ of the equation $a a_* - b b_* = 1$ such that $\xi(b) = \xi$
    and then, by Theorem \ref{t1}, there exist a unique $q \in \cP$ such that $b(\cdot) = b(\cdot,q)$.
    Moreover, it follows from Lemma \ref{hll5} that $a(\cdot,q) \in \cA$ is uniquely determined by
    $b(\cdot,q) \in \cB$ as a solution of the equation $a a_* - b b_* = 1$.
    Hence, we have $a(\cdot) = a(\cdot,q)$.
\end{proof}

\begin{proof}[\bf Proof of Theorem \ref{t11}]
    Due to Theorem \ref{hlt1}, the mappings $q \mapsto r_{\pm}(\cdot,q)$ are homeomorphisms
    between $\cL$ and $\cR_{\bu}$. Since the metrics on $\cP$ and $\cL$ and on $\cR^{\pm}$ and $\cR_{\bu}$
    are equivalent, we only need to prove that the restrictions of these mappings on $\cP$ are
    bijections between $\cP$ and $\cR^{\pm}$.

    Let $q \in \cP$. Using Theorems \ref{t1} and \ref{t2}, we have $b = b(\cdot,q) \in \cB$ and $a = a(\cdot,q) \in \cA$.
    Since $r_+(\cdot,q) = -\frac{\ol b}{a}$ and $r_-(\cdot,q) = \frac{b}{a}$, it follows from
    the definitions of $\cR^{\pm}$ that $r_{\pm}(\cdot, q) \in \cR^{\pm}$.

    Let $r_{\pm} \in \cR^{\pm}$. Then, by the definitions of $\cR^{\pm}$,
    there exist $a \in \cA$ and $b \in \cB$ such that $|a|^2 - |b|^2 =1$ and
    $r_+ = -\frac{\ol b}{a}$ (or $r_- = \frac{b}{a}$).
    Using Theorems \ref{t1} and \ref{t2}, we get that there exists a unique $q \in \cP$ such that
    $b = b(\cdot,q)$ and $a = a(\cdot,q)$. Moreover, due to Lemmas \ref{p3l4} and \ref{p3l5},
    $a$ and $b$ are uniquely determined by $r_+$ (or $r_-$). Hence, we get $r_{\pm} = r_{\pm}(\cdot,q)$.
\end{proof}

\begin{proof}[\bf Proof of Corollary \ref{t3}]
    i) By Theorem \ref{t1}, each $q \in \cP$ is uniquely determined by
    $b = b(\cdot,q) \in \cB$. Hence, we get $\hat b \in \cP$ and it follows from
    the Paley-Wiener Theorem that $b \in \cE_{Cart}$ and it satisfies
    (\ref{p2e13} -- \ref{p2e15}).
    Due to (\ref{p2e13}), $b \in \cB$ is uniquely determined by its zeros,
    by the multiplicity $p$ of the zero $k = 0$ of $b$, and
    by $b^{(p)}(0) \in \C \sm \{ 0 \}$.

    ii) By Theorem \ref{t2}, each $q \in \cP$ is uniquely determined by
    $a = a(\cdot,q) \in \cA$ and by $\xi = \xi(b(\cdot,q)) \in \Xi(B)$, where $B = a a_*-1$.
    Hence, we get $a = 1 + \cF h$ for some $h \in \cP$ and it follows from the Paley-Wiener Theorem
    that $a \in \cE_{Cart}$ and it satisfies (\ref{p2e13} -- \ref{p2e15}).
    Recall that if $a \in \cA$, then $\t_+(a) = 0$,
    $\t_-(a) = 2\g$, $a(0) \neq 0$, and
    $a(k) = 1 + o(1)$ as $k \to \pm \iy$. Thus, using (\ref{p2e13}), we see that $a \in \cA$
    is uniquely determined by its zeros.
\end{proof}

\begin{proof}[\bf Proof of Corollary \ref{c2}]
    By Theorem \ref{t2}, the mapping $q \mapsto (a,\xi)$ is a bijection between $\cP$ and
    $\cA \ts \Xi(B)$, where $a = a(\cdot,q)$, $\xi = \xi(b(\cdot,q))$, and $B = aa_* - 1$.
    Let $q_o \in \cP$. We consider the restriction of the mapping
    $q \mapsto (a,\xi)$ onto $\Iso(q_o)$. It follows from the definition of $\Iso(q_o)$ that
    $a(\cdot,q) = a(\cdot,q_o)$ for any $q \in \Iso(q_o)$ and then the mapping $q \mapsto \xi$
    is a bijection between $\Iso(q_o)$ and $\Xi(B)$, where $B = a(\cdot,q_o)a_*(\cdot,q_o) - 1$.
\end{proof}

\begin{proof} [\bf Proof of Theorem \ref{t14}]
    Let $a = a(\cdot,q) = a(\cdot,q^o)$ for some $q,q_o \in \cP$. Then we have
    $b_1 = b(\cdot,q) \in \cB$ and $b_o = b(\cdot,q^o) \in \cB$ are solutions of the equation
    $a a_* - b b_* = 1$. Let $(\vk_n)_{n \geq 1}$ be zeros of $B = b_o (b_o)_* = b_1 (b_1)_*$ in
    $\C \sm \{ 0 \}$ counted with multiplicity. We introduce
    $$
        \z_n =
        \begin{cases}
            \vk_n, & \xi_n(b_o) \geq 0\\
            \ol{ \vk_n}, & \xi_n(b_o) < 0
        \end{cases},\qq
        z_n =
        \begin{cases}
            \vk_n, & \xi_n(b_1) \geq 0\\
            \ol{ \vk_n}, & \xi_n(b_1) < 0
        \end{cases},\qq n \geq 1.
    $$
    Thus, $(\z_n)_{n \geq 1}$ are zeros of $b_o$ in $\C \sm \{ 0 \}$ and
    $(z_n)_{n \geq 1}$ are zeros of $b_1$ in $\C \sm \{ 0 \}$ counted with multiplicity.
    Using (\ref{p7e8}), we obtain
    $$
        \begin{aligned}
            b_o(k) &= \xi_0(b_o) |C|^{1/2} k^{p} e^{i\g k} \lim_{r \to +\iy}
            \prod_{|\z_n| \leq r} \left(1 - \frac{k}{\z_n}\right),\qq k \in \C,\\
            b_1(k) &= \xi_0(b_1) |C|^{1/2} k^{p} e^{i\g k} \lim_{r \to +\iy}
            \prod_{|z_n| \leq r} \left(1 - \frac{k}{z_n}\right),\qq k \in \C,
        \end{aligned}
    $$
    for some $C \in \C$ and $p \geq 0$, which yields
    $$
        P(k) = \frac{b_1(k)}{b_o(k)} = e^{i\a} \lim_{r \to +\iy} \prod_{z_n \in G,\, |z_n| < r}
        \left(1 - \frac{k}{z_n} \right) \left( 1 - \frac{k}{\ol{z_n}} \right)^{-1},\qq k \in \C,
    $$
    where $e^{i\a} = \xi_0(b_1)/\xi_0(b_o)$ and $G = \{ \, z_n,\, n \geq 1 \, \mid \, \xi_n(b_o) \neq \xi_n(b_1) \, \}$.

    Let $q_o \in \cP$ and let $(z_n)_{n \geq 1}$ be zeros of $b(\cdot,q_o)$ in $\C \sm \{ 0 \}$
    counted with multiplicity. Let $b(\cdot,q) = b(\cdot,q_o) P$ for some $q \in \cP$,
    where $P$ is given by (\ref{p1e3}) for some $\a \in \R$ and non-real subsequence $G \ss (z_n)_{n \geq 1}$.
    Recall that $|P(k)| = 1$ for any $k \in \R$.
    Hence, we have $|b(k,q)| = |b(k,q_o)|$ for any $k \in \R$. Due to Lemma \ref{p3l6}, $a(\cdot,q)$
    is uniquely determined by $|b(k,q)|$, $k \in \R$, and then $a(\cdot,q) = a(\cdot,q_o)$.
    Therefore, we have $q \in \Iso(q_o)$.
\end{proof}

\begin{proof}[\bf Proof of Theorem \ref{t12}]
    Let $q \in \cP$. Then $a = a(\cdot,q) \in \cE_{Cart}$ and it has the Hadamard factorization (\ref{p7e3}):
    $$
        a(k) = e^{i\phi}|C|^{1/2} e^{i\g k} \lim_{r \to +\iy}
        \prod_{|k_n| \leq r} \left(1 - \frac{k}{k_n}\right),\qq k \in \C
    $$
    for some $\phi \in \R$ and $|C| = |a(0)|^2$, where the product converges uniformly on compact subsets of $\C$.
    Hence, we get
    $$
        |a(k)| = |a(0)| \lim_{r \to +\iy}
        \prod_{|k_n| \leq r} \left|1 - \frac{k}{k_n}\right|,\qq k \in \C,
    $$
    which yields that
    \[ \label{p4e21}
        \log |a(k)| = \log|a(0)| + \lim_{r \to +\iy}
        \sum_{|k_n| \leq r} \log \left|1 - \frac{k}{k_n}\right| + 2\pi i N,\qq k \in \C,
    \]
    for some $N \in \Z$, where we fixed the branch of the logarithm by $\log x \in \R$ for $x \in \R_+$
    (see, e.g., p.~16 in \cite{T58}).
    Substituting $k = 0$ in (\ref{p4e21}), we get $N = 0$ and then we obtain (\ref{p1e19}).

    Let $(z_n)_{n \geq 1}$ be zeros of $b = b(\cdot,q)$ in $\C \sm \{ 0 \}$. Let also $\xi = \xi(b)$
    and let $p$ be the multiplicity of the zero $k = 0$ of $b(k)$.
    Recall that $b \in \cE_{Cart}$ and it has the Hadamard factorization (\ref{p7e8}):
    \[ \label{p4e23}
        b(k) = \xi_0 |C|^{1/2} k^{p} e^{i\g k} \lim_{r \to +\iy}
        \prod_{|z_n| \leq r} \left(1 - \frac{k}{z_n}\right),\qq k \in \C
    \]
    for some $C \in \R$, where the product converges uniformly on compact subsets of $\C$.
    Let $I \ss \R$ be an open interval such that $b(k) \neq 0$ for any $k \in I$. Then we have
    $$
        \int_{k_1}^{k_2} \frac{b'(s)}{b(s)} ds = \ln b(k_2) - \ln b(k_1),\qq k_1,k_2 \in I,
    $$
    which yields
    \[ \label{s4e1}
        \arg b(k_2) = \arg b(k_1) + \int_{k_1}^{k_2} \Im \frac{b'(s)}{b(s)} ds,\qq k_1,k_2 \in I.
    \]
    Using (\ref{p4e23}), we obtain
    \[ \label{s4e3}
        \Im \frac{b'(s)}{b(s)} = \g + \sum_{n = 1}^{\iy} \xi_n \frac{|\Im z_n|}{|s - z_n|^2} = \g + w(s),\qq
        s \in \R \sm (z_n)_{n \geq 1},
    \]
    where the series converges absolutely and uniformly on compact subsets of
    $\R \sm (z_n)_{n \geq 1}$, since (\ref{p7e7}) holds true.

    Using (\ref{s4e1}) and (\ref{s4e3}),
    we can calculate $\arg b$ between its zeros on $\R$. Since $b$ is entire, it has finitely many
    zeros on any compact interval. Now, we describe how $\arg b$ changes in neighborhood of its zero.
    Let $z_o \in \R$ be a zero of $b$ with the multiplicity $n$. Then we have $b(k) = C(k-z_o)^n + o(k)$
    as $k \to z_o$. Using this asymptotics, we obtain
    \[ \label{s4e2}
        \lim_{\ve \to +0} (\arg(b(k+\ve) - \arg(b(k-\ve))) = -\pi n.
    \]
    Combining (\ref{s4e1}) and (\ref{s4e2}), we get
    $$
        \arg b(k) = \arg b(0) + \g k - \pi I(k) + \int_0^k w(s) ds,\qq k \in \R \sm (z_n)_{n \geq 1},
    $$
    where $I$ and $w$ are given by (\ref{p1e21}). Finally, if $b(0) \neq 0$, then it follows from
    (\ref{p4e23}) that $\arg b(0) = \arg \xi_0$, which yields (\ref{p1e20}).
\end{proof}

\begin{proof}[\bf Proof of Theorem \ref{t5}]
    Let $q \in \cP$. By Theorem \ref{t2}, $a = a(\cdot,q) \in \cA$ and
    then there exists $h \in \cP$ such that $a = 1 + \cF h$. It is well-known that
    the set of smooth compactly supported functions $C_o^{\iy}(0,\g)$ is dense in $L^2(0,\g)$.
    Thus, for any $\ve > 0$, there exists $h_1 \in C_o^{\iy}(0,\g)$ such that $h = h_1 + h_2$ and
    $\| h_2 \|_{L^2(0,\g)} < \ve |\g|^{-1/2}$.
    Let $a(k) = 0$ for some $k \in \C_-$. Then we have $\cF h(k) = -1$.
    Estimating the left-hand side of this identity, we get
    \[ \label{p4e18}
        |\cF h_1(k)| + |\cF h_2(k)| \geq 1.
    \]
    Since $h_1 \in C_o^{\iy}(0,\g)$, we obtain
    \[ \label{p4e19}
        \begin{aligned}
            |\cF h_1(k)| &\leq \left| \int_0^{\g} h_1(s) e^{2iks} ds \right| =
            \left| \frac{1}{2k}\int_0^{\g} h'_1(s) e^{2iks} ds \right| \\
            &\leq \frac{1}{2|k|} \int_0^{\g} |h'_1(s)| e^{-2s \Im k} ds \leq
            \frac{e^{-2\g \Im k}}{2|k|} \|h'_1 \|_{L^1(0,\g)} = C e^{-2 \g \Im k}\frac{1}{|k|}.
        \end{aligned}
    \]
    Using $\| h_2 \|_{L^1(0,\g)} \leq \sqrt{\g} \| h_2 \|_{L^2(0,\g)} =
    \ve$, we have
    \[ \label{p4e20}
        |\cF h_2(k)| \leq \int_0^{\g} |h_2(s)| e^{-2s\Im k} ds \leq
        e^{-2 \g \Im k} \| h_2 \|_{L^1(0,\g)} = \ve e^{-2 \g \Im k}.
    \]
    Substituting (\ref{p4e19}) and (\ref{p4e20}) in (\ref{p4e18}), we get
    $$
        e^{-2\g \Im k} \left(\ve + \frac{C}{|k|} \right) \geq 1,
    $$
    which yields (\ref{p1e1}). Now we consider (\ref{p1e1}) for $|k_n| \to \iy$.
    For fixed $\ve > 0$ and $C \geq 0$, we~get
    $$
        2\g \Im k_n \leq \ln(\ve) + O(|k|^{-1}).
    $$
    Thus, there are finitely many resonances such that $\Im k_n > \ln(\ve)$.
    Since it holds for any $\ve > 0$, we complete the proof of the theorem.
\end{proof}

\begin{proof}[\bf Proof of Theorem \ref{t4}]
    For simplicity, we consider the case when $N = 1$. We introduce
    $$
        B(k) = \frac{k-z_1}{k-z_0},\qq R(k) = \frac{1}{k-z_0},\qq k \in \C \sm \{z_0\}.
    $$
    Firstly, we show that $b_1 = Rb \in \cB$, where $b(\cdot) = b(\cdot,q^o)$ for some $q^o \in \cP$.
    Since $b \in \cE_{Cart}$ and $b(z_0) = 0$,
    it follows that $b_1$ is entire. Using the definition of $\t_{\pm}$, we have
    $$
        \t_{\pm}(b_1) = \t_{\pm}(b) + \t_{\pm}(R) = \t_{\pm}(b),
    $$
    where $\t_{\pm}(R) = 0$, since $\log(R(\pm i y)) \to -\log y$ as $y \to +\iy$. Thus, we get $b_1 \in \cE_{Cart}$.
    Now, we show that $b_1 \in L^2(\R)$. Since $b_1 \in \cE_{Cart}$, it follows that, for any $\ve > 0$,
    there exists $C > 0$ such that $|b_1(k)| < C$ for each $|k-z_0| < \ve$. Note that
    $|R(k)| < 1/\ve$ for each $|k-z_0| > \ve$. Using these estimates, we get
    $$
        \int_{\R} |b_1(s)|^2 ds = \int_{\R \sm |k-z_0| < \ve} |b_1(s)|^2 ds +
        \int_{\R \cap |k-z_0| < \ve} |b_1(s)|^2 ds \leq \frac{1}{\ve^2} \int_{\R} |b(s)|^2 ds + 2 \ve C^2 < +\iy,
    $$
    which yields $b_1 \in L^2(\R)$. It follows from the Paley-Wiener Theorem
    that $b_1 = \cF g$
    for some $g \in \cP$ and then we get $b_1 \in \cB$. Thus, by Theorem \ref{t1}, there exists a unique
    $p \in \cP$ such that $b_1 = b(\cdot,p)$.

    Secondly, we show that $b_2 = Bb \in \cB$. As above, we get $Bb \in \cE_{Cart}$. Now, we show that
    $b_2 \in L^2(\R)$. Using $Rb \in L^2(\R)$, we obtain
    \[ \label{p5e11}
        \|b_2 - b\|_{L^2(\R)} = \|Bb - b\|_{L^2(\R)} = \|(z_1-z_0)Rb\|_{L^2(\R)} \leq |z_1 - z_0| \|Rb\|_{L^2(\R)} < +\iy.
    \]
    Using (\ref{p5e11}) and $b \in L^2(\R)$, we get $b_2 \in L^2(\R)$ and then we have $b_2 \in \cB$.
    Thus, it follows from Theorem
    \ref{t1} that there exists a unique $q \in \cP$ such that $b_2(\cdot) = b(\cdot,q)$.

    Thirdly, we show that $\rho_{\cB}(b_2,b) \to 0$ as $z_1 \to z_0$. Using (\ref{p2e10}),
    (\ref{p5e11}), and the Plancherel theorem (see e.g. Theorem IX.6 in \cite{RS80}), we get
    $$
        \rho_{\cB}(b_2,b) = \| b_2 - b \|_{L^2(\R)} \leq |z_1 - z_0| \|Rb\|_{L^2(\R)}.
    $$
    Thus, it follows that $\rho_{\cB}(b_2,b) \to 0$ as $z_1 \to z_0$.
    By Theorems \ref{t1} and \ref{hlt1}, the mapping
    $q \mapsto b(\cdot,q)$ is a homeomorphism between $\cP$ and $\cB$ and the mapping
    $q \mapsto r_{\pm}(\cdot,q)$ is a homeomorphism between $\cP$ and $\cR^{\pm}$. Thus, we have
    $\rho_{\cP}(q_o,q) \to 0$ and $\rho_{\cR}(r_{\pm}(\cdot,q_o),r_{\pm}(\cdot,q)) \to 0$ as $z_1 \to z_0$.
\end{proof}

\begin{proof} [\bf Proof of Theorem \ref{t8}]
    Let $a_o = a(\cdot,q^o) \in \cA_{symm}$ for some $q^o \in \cP$ and let $a_o(k_o) = 0$ for some
    $k_o \in \C_-$. Let $k_1 \in \C_-$. We introduce
    $$
        S(k) = \frac{(k - k_1)(k + \ol k_1)}{(k - k_o)(k + \ol k_o)} =
        \left(1 + \frac{c}{k-k_o}\right)\left(1 + \frac{\ol c}{k-\ol{k_o}}\right),\qq k \in \C,
    $$
    where $c = k_1 - k_o$. Due to $k_o \in \C_{-}$, we have $S \in L^{\iy}(\R) \cap C(\R)$.
    We establish other properties of $S$.
    By direct calculation, we get $S_*(k) = S(-k)$, $k \in \C$. Thus, we have
    $$
        \begin{aligned}
            |S(k)|^2 &= S(k) S(-k) =
            \left(1 + \frac{c}{k-k_o}\right)\left(1 + \frac{\ol c}{k-\ol{k_o}}\right)
            \left(1 - \frac{c}{k+k_o}\right)\left(1 - \frac{\ol c}{k+\ol{k_o}}\right) \\
            &= \left(1 - \frac{c(k_1+k_o)}{k^2-k_o^2}\right)
            \left(1 - \frac{\ol c (\ol{k_1} + \ol{k_o})}{k^2-\ol{k_o^2}}\right) =
            \left(1 - \frac{\vk}{E^2-E_o^2}\right)
            \left(1 - \frac{\ol{\vk}}{E^2-\ol{E_o^2}}\right)\\
            &= 1 - \frac{2 E \Re \vk - 2 \Re(\vk \ol{E_o}) - |\vk|^2}{|E-E_o|^2} =
            1 - \frac{2 E \Re \vk -|k_1|^4 + |k_o|^4}{|E-E_o|^2}\\
            &= 1 - G(E),
        \end{aligned}
    $$
    where $\vk = k_1^2 - k_o^2$, $E = k^2$, and $E_o = k_o^2$. Using (\ref{p6e10}), we have
    $G(E) \leq 0$, $E \geq 0$, and then $|S(k)| \geq 1$ for any $k \in \R$. Since $G(E) = O(E^{-1})$ as $E \to \iy$,
    we have $|S(k)|^2 - 1 = O(k^{-2})$ as $k \to \iy$. Due to $S \in L^{\iy}(\R) \cap C(\R)$,
    it follows that $|S|^2 - 1 \in L^1(\R)$. Now, we consider $S-1$:
    \[ \label{p6e17}
        \begin{aligned}
            S(k) - 1 &= \frac{(k - k_1)(k + \ol k_1) - (k - k_o)(k + \ol k_o)}{(k - k_o)(k + \ol k_o)}\\
            &= \frac{2ik(\Im k_o - \Im k_1) + |k_o|^2 - |k_1|^2}{k^2 - 2ik\Im k_o - |k_o|^2},\qq k \in \C.
        \end{aligned}
    \]
    Using (\ref{p6e17}), we have $S(k) = 1 + O(k^{-1})$ as $k \to \iy$.
    Since $S \in L^{\iy}(\R) \cap C(\R)$, it follows that $\| S - 1 \|_{L^2(\R)} < \iy$.

    Now we show that $a = Sa_o \in \cA_{symm}$.

    1) Since $a_o(k_o) = 0$, $k_1 \in \C_-$, and $a_o \in \cA_{symm}$, it follows that
    $a(-\ol{k_o}) = 0$ and then $a$ is entire and $a(k) \neq 0$ for any $k \in \C_+$.
    As in proof of Theorem \ref{t4}, we obtain $\t_{\pm}(S) = 0$, which yields
    $\t_{\pm}(a) = \t_{\pm}(a_o)$. Due to $S, a_o \in L^{\iy}(\R)$, it follows that $a \in \cE_{Cart}$.

    2) Since $|S(k)| \geq 1$ and $|a_o(k)| \geq 1$ for any $k \in \R$, we have $|a(k)| \geq 1$
    for any $k \in \R$.

    3) By direct calculation, we get
    \[ \label{p6e1}
        |a|^2 - 1 = (|a_o|^2 - 1)|S|^2 + |S|^2 - 1.
    \]
    Since $|a_o|^2 - 1 \in L^1(\R)$, $|S|^2 - 1 \in L^1(\R)$, and $S \in L^{\iy}(\R)$, it follows
    from (\ref{p6e1}) that $|a|^2 - 1 \in L^1(\R)$.

    4) We have
    $$
        \| a - 1\|_{L^2(\R)} = \| (a_o - 1)S + S-1\|_{L^2(\R)} \leq
        \| S \|_{L^{\iy}(\R)} \| a_o - 1 \|_{L^2(\R)}  + \| S - 1 \|_{L^2(\R)}.
    $$
    Since $a_o - 1 \in L^2(\R)$, $S - 1 \in L^2(\R)$, and $S \in L^{\iy}(\R)$, we have
    $a - 1 \in L^2(\R)$. Recall that $a-1 \in \cE_{Cart}$. Then, using the Paley-Wiener Theorem,
    we get $a = 1 + \cF h$ for some $h \in \cP$.

    Thus, we have $a \in \cA$. Due to $S*(k) = S(-k)$ and $a*(k) = a(-k)$ for any $k \in \C$,
    it follows that $a_1 \in \cA_{symm}$. As we noted above, $\Xi(a a_* - 1) \neq \es$. Then,
    by Theorem \ref{t2}, there exists $q \in \cP$ such that $a = a(\cdot,q)$.
\end{proof}

\section{Symmetries of the potentials} \label{p5}
In this section, we solve the inverse problem for the Dirac operators with potentials, which have
some symmetries. Namely, we consider real-valued, even, and odd potentials. Firstly, we consider
some simple transformations of the potential and show how the scattering data changes under such transformations.
We give the following simple lemma.
\begin{lemma} \label{p9l1}
    Let $a = a(\cdot,q)$, $b = b(\cdot,q)$, $r_{\pm} = r_{\pm}(\cdot,q)$ for some $q \in \cL$
    and let $p \in \cL$.
    Then the following statements hold true:
    \begin{enumerate}[label={\roman*)}]
        \item We have $p(x) = q(-x)$, $x \in \R$, if and only if
        $$
            a(k,p) = \ol{a(-k)},\qq b(k,p) = b(-k,q),\qq r_{\pm}(k,p) = - \ol{r_{\mp}(-k)},\qq k \in \R.
        $$
        \item We have $p = \ol q$ if and only if
        $$
            a(k,p) = \ol{a(-k)},\qq b(k,p) = \ol{b(-k)},\qq r_{\pm}(k,p) = \ol{r_{\pm}(-k)},\qq k \in \R.
        $$
        \item Let $\a \in \R$. Then we have $p = e^{i\a}q$ if and only if
        $$
            a(k,p) = a(k),\qq b(k,p) = e^{-i\a} b(k),\qq r_{\pm}(k,p) = e^{\pm i\a} r_{\pm}(k),\qq k \in \R.
        $$
        \item Let $s \in \R$. Then we have $p(x) = q(x+s)$, $x \in \R$, if and only if
        $$
            a(k,p) = a(k),\qq b(k,p) = e^{-2iks} b(k),\qq r_{\pm}(k,p) = e^{\pm 2iks} r_{\pm}(k),\qq k \in \R.
        $$
        \item Let $s \in \R$. Then we have $p(x) = e^{2isx}q(x)$, $x \in \R$, if and only if
        $$
            a(k,p) = a(k-s),\qq b(k,p) = b(k-s),\qq r_{\pm}(k,p) = r_{\pm}(k-s),\qq k \in \R.
        $$
    \end{enumerate}
\end{lemma}
\begin{remark}
    1) Recall that $\frac{1}{\pi} \log |a(k,q)|$ is the action variable for the defocusing NLS equation.
    Due to iii) and iv), the action variable does not change when the potential is shifted
    along the real line or multiplied by the phase factor $e^{i\a}$.

    2) Above, we consider potentials from $\cP$ such that its support is $[0,\g]$ for some $\g > 0$.
    Using iv), we can easily adapt these results to the case when the support is $[d,d+\g]$ for any
    $d \in \R$ and $\g > 0$.
\end{remark}
\begin{proof}
    i) Let $p(x) = q(-x)$, $x \in \R$, and let $k \in \R$. We consider
    \[ \label{p9e1}
        h^{\pm}(x) = \s_3 f^{\mp}(-x,-k,q)\s_3,\qq x \in \R.
    \]
    Differentiating $h^{\pm}(x)$ by $x$, we get
    $$
        \begin{aligned}
            (h^{\pm})'(x) &= -\s_3(f^{\mp})'(-x,-k,q)\s_3\\
            &= -\s_3Q(-x)f^{\mp}(-x,-k,q)\s_3 +
            \s_3 ik\s_3 f^{\mp}(-x,-k,q)\s_3\\
            &= Q(-x)h^{\pm}(x) + ik\s_3h^{\pm}(x),\qq x \in \R.
        \end{aligned}
    $$
    Here we used $-\s_3Q(-x)\s_3 = Q(-x)$. Thus, it follows
    that $h^{\pm}$ are solutions of the Dirac equation for the potential $p$. Moreover, it follows
    from (\ref{p9e1}) that $h^{\pm}(x) = e^{ikx\s_3}(1+o(1))$ as $x \to \pm \iy$, which yields
    $h^{\pm}(x) = f^{\mp}(x,k,p)$, $(x,k) \in \R^2$. Substituting (\ref{p9e1}) in (\ref{p3e17}),
    we get
    \[ \label{p9e2}
        \begin{aligned}
            A(k,p) &= (f^{-}(x,k,p))^{-1} f^{+}(x,k,p)\\
            &= \s_3 (f^{+}(-x,-k,q))^{-1}\s_3 \s_3 f^{-}(-x,-k,q)\s_3 = \s_3 A(-k,q)^{-1} \s_3.
        \end{aligned}
    \]
    Substituting (\ref{p3e25}) in (\ref{p9e2}), we obtain
    $$
        a(k,p) = \ol{a(-k)},\qq b(k,p) = b(-k),\qq k \in \R,
    $$
    and then, using $r_+ = - \frac{\ol b}{a}$, $r_- = \frac{b}{a}$, we have
    $r_{\pm}(k,p) = - \ol{r_{\mp}(-k)}$, $k \in \R$.

    Let $p \in \cL$ be such that $r_{\pm}(k,p) = - \ol{r_{\mp}(-k)}$, $k \in \R$. Then we have
    for any $s \in \R$
    \[ \label{p9e3}
        \begin{aligned}
            F_{\pm}(s,p) &= \hat{r}_{\pm}(s,p) = \frac{1}{\pi} \int_{\R} r_{\pm}(k,p) e^{-2iks} dk \\
            &=-\frac{1}{\pi} \int_{\R} \ol{r_{\mp}(-k)} e^{-2iks} dk =
            -\frac{1}{\pi} \int_{\R} \ol{r_{\mp}(z)} \ol{e^{-2izs}} dz = -\ol{F_{\mp}(s)},
        \end{aligned}
    \]
    where $F_{\pm} = F_{\pm}(\cdot,q)$. Here we used the change of variables $z = -k$.
    Substituting (\ref{p9e3}) in (\ref{hlF_scat}), we obtain
    \[ \label{p9e4}
        \Omega_{\pm}(s,p) = \s_3 \Omega_{\mp}(-s) \s_3,\qq s \in \R,
    \]
    where $\Omega_{\pm} = \Omega_{\pm}(\cdot,q)$.
    It follows from (\ref{p3e15}) and Lemma \ref{hll3} that for any sign $\pm$ there exists
    a unique solution $\G^{\pm}(x,s)$ of the GLM equation
    \[ \label{p5e7}
        \G^{\pm}(x,s) + \Omega^{\pm}(x+s) \pm \int_0^{\pm \iy} \G^{\pm}(x,t) \Omega^{\pm}(x+t+s) dt = 0,
        \qq (x,s) \in \R \ts \R_{\pm},
    \]
    such that $\G^{\pm}_{12}(x,0) = \mp q(x)$ for almost all $x \in \R$. We introduce
    \[ \label{p5e8}
        \G^{\pm}_o(x,s) = \s_3 \G^{\mp}(-x,-s) \s_3,\qq (x,s) \in \R \ts \R_{\pm}.
    \]
    Substituting (\ref{p9e4}) in (\ref{p5e7}) and using (\ref{p5e8}), we get for almost all
    $(x,s) \in \R \ts \R_{\pm}$
    $$
        \G^{\pm}_{o}(x,s) + \Omega^{\pm}(x+s,p) \pm \int_0^{\pm \iy} \G^{\pm}_{o}(x,t) \Omega^{\pm}(x+t+s,p) dt = 0.
    $$
    Thus, it follows from Lemma \ref{hll3} that $\G^{\pm}_{o}(x,s) = \G^{\pm}(x,s,p)$
    and then, due to (\ref{p3e15}) and (\ref{p5e8}), we get for almost all $x \in \R$
    $$
        p(x) = \mp (\G^{\pm}_{o})_{12}(x,0) = \pm \G^{\mp}_{12}(-x,0) = q(-x).
    $$
    In the other cases, the proof is similar. We only show how $f^{\pm}$, $\G^{\pm}$, $F_{\pm}$, and
    $\Omega_{\pm}$ are transformed in each case. These formulas can be proved by direct calculations.\\
    ii) If $p = \ol q$, then we get for any $x,k,s \in \R$ and $t \in \R_{\pm}$
    $$
        \begin{aligned}
            f^{\pm}(x,k,p) &= \s_1 f^{\pm}(x,-k,q)\s_1,\qq
            \G^{\pm}(x,t,p) = \s_1 \G^{\pm}(x,t,q) \s_1,\\
            F_{\pm}(s,p) &= \ol{F_{\pm}(s,q)},\qq
            \Omega_{\pm}(s,p) = \s_1 \Omega_{\pm}(s,q) \s_1.
        \end{aligned}
    $$
    iii) If $p = e^{i\a}q$ for some $\a \in \R$, then we get for any $x,k,s \in \R$ and $t \in \R_{\pm}$
    $$
        \begin{aligned}
            f^{\pm}(x,k,p) &= e^{i\frac{\a}{2}\s_3} f^{\pm}(x,k,q)e^{-i\frac{\a}{2}\s_3},\qq
            \G^{\pm}(x,t,p) = e^{i\frac{\a}{2}\s_3} \G^{\pm}(x,t,q) e^{-i\frac{\a}{2}\s_3},\\
            F_{\pm}(s,p) &= e^{\pm i \a} F_{\pm}(s,q),\qq
            \Omega_{\pm}(s,p) = e^{i\frac{\a}{2}\s_3} \Omega_{\pm}(s,q) e^{-i\frac{\a}{2}\s_3}.
        \end{aligned}
    $$
    iv) If $p(x) = q(x+d)$ for some $d \in \R$, then we get for any $x,k,s \in \R$ and $t \in \R_{\pm}$
    $$
        \begin{aligned}
            f^{\pm}(x,k,p) &= f^{\pm}(x+d,k,q)e^{-ikd\s_3},\qq
            \G^{\pm}(x,t,p) = \G^{\pm}(x+d,t,q),\\
            F_{\pm}(s,p) &= F_{\pm}(s-d,q),\qq
            \Omega_{\pm}(s,p) = \Omega_{\pm}(s\pm d,q).
        \end{aligned}
    $$
    v) If $p(x) = e^{2idx}q(x)$ for some $d \in \R$, then we get for any $x,k,s \in \R$ and $t \in \R_{\pm}$
    $$
        \begin{aligned}
            f^{\pm}(x,k,p) &= e^{idx\s_3} f^{\pm}(x,k-d,q),\qq
            \G^{\pm}(x,t,p) = e^{idx\s_3} \G^{\pm}(x,t,q) e^{-i(2t+x)d\s_3},\\
            F_{\pm}(s,p) &= e^{-2isd\s_3} F_{\pm}(s,q),\qq
            \Omega_{\pm}(s,p) = e^{isd\s_3} \Omega_{\pm}(s,q) e^{-isd\s_3}.
        \end{aligned}
    $$
\end{proof}
Now, we introduce the subclasses of potentials, which have some symmetries:
$$
    \begin{aligned}
        \cP_{even} &= \{\, q \in \cP \,\mid\, q(x) = q(\g-x),\, x \in \R \,\},\\
        \cP_{odd} &= \{\, q \in \cP \,\mid\, q(x) = -q(\g-x),\, x \in \R \,\},\\
        \cP_{real} &= \{\, q \in \cP \,\mid\, q(x) \in \R,\, x \in \R \,\}.
    \end{aligned}
$$
\begin{remark}
    Note that $\cP_{even}$, $\cP_{odd}$, and $\cP_{real}$ equipped with the metric $\r_{\cP}$
    given by (\ref{p1e16}) are incomplete metric spaces.
\end{remark}
We also introduce the associated classes of the coefficients $b$:
$$
    \begin{aligned}
        \cB_{even} &= \{\, b \in \cB \, \mid \, b(k) = e^{2i\g k} b(-k),\, k \in \C \,\},\\
        \cB_{odd} &= \{\, b \in \cB \, \mid \, b(k) = -e^{2i\g k} b(-k),\, k \in \C \,\},\\
        \cB_{real} &= \{\, b \in \cB \, \mid \, b(k) = b_*(-k),\, k \in \C \,\}.
    \end{aligned}
$$
\begin{remark}
    1) We can give another factorization of these classes in terms of the Fourier transform:
    \begin{enumerate}[label={(\roman*)}]
        \item $b \in \cB_{even}$ if and only if $\hat b(s) = \hat b(\g-s)$, $s \in \R$;
        \item $b \in \cB_{odd}$ if and only if $\hat b(s) = -\hat b(\g-s)$, $s \in \R$;
        \item $b \in \cB_{real}$ if and only if $\hat b = \ol{\hat b}$.
    \end{enumerate}

    2) Note that $\cB_{even}$, $\cB_{odd}$, and $\cB_{real}$ equipped with the metric $\r_{\cB}$
    given by (\ref{p2e10}) are incomplete metric spaces.

    3) If $b \in \cB_{even} \cup \cB_{odd}$, then $z$ is a zero of $b$ of multiplicity $n$
    if and only if $-z$ is a zero of $b$ of multiplicity $n$. Moreover, below we show that
    if $b \in \cB_{even}$ (or $b \in \cB_{odd}$), then the zero $z = 0$ of $b$ has even (or odd) multiplicity.
    If $b \in \cB_{real}$, then
    $z$ is a zero of $b$ of multiplicity $n$
    if and only if $-\ol{z}$ is a zero of $b$ of multiplicity $n$.
\end{remark}
Using Theorem \ref{t1}, we obtain the following characterization of the potentials with symmetries.
\begin{corollary} \label{p9c1}
    The restriction of the mapping $q \mapsto b(\cdot,q)$ on $\cP_{\a}$ is a homeomorphism
    between $\cP_{\a}$ and $\cB_{\a}$ for each $\a \in \{\, even,\, odd,\, real\,\}$.
\end{corollary}
\begin{proof}
    Due to Theorem \ref{t1}, the mapping $q \mapsto b(\cdot,q)$ is a homeomorphism between $\cP$ and $\cB$
    and then $q = p$ if and only if $b(\cdot,q) = b(\cdot,p)$.

    Let $p(x) = q(\g-x)$, $x \in \R$. It follows from Lemma \ref{p9l1} i) and iv) that
    $b(k,p) = e^{2ik\g} b(-k,q)$, $k \in \R$.
    Thus, $q \in \cP_{even}$ if and only if $b \in \cP_{even}$.

    Let now $p(x) = e^{i\pi}q(\g-x)$, $x \in \R$.
    Using Lemma \ref{p9l1} i), iii), and iv), we get $b(k,p) = - e^{2ik\g} b(-k,q)$,$k \in \R$.
    Thus, $q \in \cP_{odd}$ if and only if $b \in \cP_{odd}$.

    Finally, let $p = \ol q$. Using Lemma \ref{p9l1} ii), we get $b(k,p) = \ol{b(-k,q)}$, $k \in \R$.
    Thus, $q \in \cP_{real}$ if and only if $b \in \cP_{real}$.
\end{proof}
\begin{remark}
    Using Theorem \ref{hlt2} and Lemma \ref{p9l1}, we can obtain characterization of the potentials
    from $\cL$, which have some symmetries.
\end{remark}
Recall that the coefficient $a(\cdot,q)$ does not uniquely determine a potential $q$ and we need information
about position of the zeros of $b(\cdot,q)$. If the potential $q$ has some symmetry, then the zeros
of $b(\cdot,q)$ also have some symmetry. In this case, we have less degrees of freedom in choosing
the position of zeros of $b(\cdot,q)$. In order to give this results formally, we introduce
special subsets of $\cA$. Recall that
$$
    \cA_{symm} = \{\, a \in \cA \, \mid \, a(k) = a_*(-k),\, k \in \C \,\}.
$$
Note that the zeros of $a \in \cA_{symm}$ are symmetric with respect to the imaginary line.
We introduce a notation: $\#(z_o,f)$ is a multiplicity of the zero $z_o$ of the function $f$
and it equals zero if $z_o$ is not a zero of $f$. We introduce the following subspaces of $\cA$:
$$
    \begin{aligned}
        \cA_{even} &= \{\, a \in \cA_{symm} \, \mid \, \#(0,aa_*-1) = 4m,\, m \geq 0 \,\},\\
        \cA_{odd} &= \{\, a \in \cA_{symm} \, \mid \, \#(0,aa_*-1) = 4m+2,\, m \geq 0 \,\},\\
        \cA_{double} &= \{\, a \in \cA_{symm} \, \mid \, \#(z,aa_*-1) = 2m,\, m \geq 0,\, \forall z \in \C \,\}.
    \end{aligned}
$$
\begin{remark}
    Note that $\cA_{symm} = \cA_{even} \cup \cA_{odd}$ and $\cA_{even} \cap \cA_{odd} = \es$.
\end{remark}
Let $B$ be an exponential type function. Recall that we have introduced $\xi(b) = (\xi_n)_{n \geq 0}$ in
(\ref{p1e18}) to parametrize the space of solutions $b \in \cB$ of the equation $bb_* = B$ such that
$$
    \xi(b) = (\xi_n)_{n \geq 0},\qq \xi_0 = \frac{b^{(p)}(0)}{|b^{(p)}(0)|},\qq
    \xi_n = \sign \Im z_n =
    \begin{cases}
        1,& \Im z_n > 0\\
        0,& \Im z_n  = 0\\
        -1,& \Im z_n < 0
    \end{cases},\qq n \geq 1,
$$
where $(z_n)_{n \geq 1}$ are zeros of $b$ in $\C \sm \{ 0 \}$ and $p$ is the multiplicity of the
zero $k = 0$ of $b(k)$. Now, we introduce the following subspaces of $\Xi(B)$:
$$
    \begin{aligned}
        \Xi_{even}(B) &= \{ \, \xi(b) \in \Xi(B) \, \mid \, b \in \cB_{even},\, b b_* = B \, \},\\
        \Xi_{odd}(B) &= \{ \, \xi(b) \in \Xi(B) \, \mid \, b \in \cB_{odd},\, b b_* = B \, \},\\
        \Xi_{real}(B) &= \{ \, \xi(b) \in \Xi(B) \, \mid \, b \in \cB_{real},\, b b_* = B \, \}.
    \end{aligned}
$$
\begin{remark}
    Below, we show that if $f = aa_* - 1$ for some $a \in \cA_{symm}$, then $\Xi_{\a}(f) \neq \es$
    for each $\a \in \{\, even,\, odd,\, real\,\}$.
\end{remark}
Using Theorem \ref{t2} and Corollary \ref{p9c1}, we obtain the following corollary.

\begin{corollary} \label{p9c2}
    Let $a = a(\cdot,q)$ and $\xi = \xi(b(\cdot,q))$ for any $q \in \cP$ and let $B = aa_*-1$.
    Then the following mappings are bijections:
    \begin{enumerate}[label={\roman*)}]
        \item $q \mapsto (a, \xi)$ from $\cP_{even}$ into $\cA_{even} \ts \Xi_{even}(B)$;
        \item $q \mapsto (a, \xi)$ from $\cP_{odd}$ into $\cA_{odd} \ts \Xi_{odd}(B)$;
        \item $q \mapsto (a, \xi)$ from $\cP_{real}$ into $\cA_{symm} \ts \Xi_{real}(B)$;
        \item $q \mapsto (a, \xi_0)$ from $\cP_{real} \cap \cP_{even}$ into
        $(\cA_{even} \cap \cA_{double}) \ts \{-1,1\}$;
        \item $q \mapsto (a, \xi_0)$ from $\cP_{real} \cap \cP_{odd}$ into
        $(\cA_{odd} \cap \cA_{double}) \ts \{-1,1\}$.
    \end{enumerate}
\end{corollary}
\begin{proof}
    Firstly, we consider the case $q \in \cP_{even}$.
    In the cases $q \in \cP_{odd}$ and $q \in \cP_{real}$, the proofs are similar.
    Let $q \in \cP_{even}$. Then, by Corollary \ref{p9c1}, we have $b = b(\cdot,q) \in \cB_{even}$,
    which yields $\xi(b) \in \Xi_{even}(B)$, where $B = b b_*$.

    Using Lemma \ref{p9l1} i) and iv), we get that $a(k) = \ol{a(-k)}$ for any $k \in \R$, where
    $a = a(\cdot,q)$. Since $q \in \cP$, it follows that $a$ is entire. Thus, we have $a(k) = a_*(-k)$,
    $k \in \R$, which yields $a \in \cA_{symm}$. Now, we show that $\#(0,a a_* - 1) = 4m$ for some
    $m \geq 0$. It is easy to see that
    $$
        \#(0,a a_* - 1) = \#(0,B) = 2\#(0,b).
    $$
    Since $b$ is entire, we have $b(k) = C k^n + o(k^n)$ as $k \to 0$ for some $C \neq 0$ and
    $n = \#(0,b)$. Using $b(k) = e^{2i\g k}b(-k)$, we get $C k^n  + o(k^n)= C (-1)^n k^n + o(k^n)$
    as $k \to 0$ and then
    $n = 2m$ for some $m \geq 0$. Thus, we have $a \in \cA_{even}$.

    Let $a \in \cA_{even}$ and $\xi \in \Xi_{even}(B)$, where $B = aa_*-1$. Since
    $\cA_{even} \ss \cA$ and $\Xi_{even}(B) \ss \Xi(B)$, it follows from Theorem \ref{t2} that
    there exists a unique $q \in \cP$ such that $a = a(\cdot,q)$ and $\xi = \xi(b)$,
    where $b = b(\cdot,q)$. Due to $\xi \in \Xi_{even}(B)$, we get $b \in \cB_{even}$. Hence, by
    Corollary \ref{p9c1}, we obtain $q \in \cP_{even}$.

    Secondly, we consider the case $q \in \cP_{even} \cap \cP_{real}$. In the case
    $q \in \cP_{odd} \cap \cP_{real}$, the proof is similar. Using the previous parts of
    the corollary, we see that $a = a(\cdot,q) \in \cA_{even}$ and
    $\xi \in \Xi_{even}(B) \cap \Xi_{real}(B)$, where $B = bb_*$ and
    $b = b(\cdot,q) \in \cB_{even} \cap \cB_{real}$. It follows that
    $$
        b(k) = b_*(k),\qq b_*(k) = b(-k),\qq k \in \C,
    $$
    i.e. the zeros of $b$ are symmetric with respect to real and imaginary lines.
    Thus, the zeros of $b$ are uniquely determined by the zeros of $B$ and the multiplicity of
    each zero of $B$ is even. This implies that $a \in \cA_{double}$ and $\xi_n(b)$, $n \geq 1$,
    are fixed. Finally, due to $b \in \cB_{real}$, we get $\xi_0(b) \in \{-1,1\}$.

    Let $a \in \cA_{double} \cap \cA_{even}$ and $\xi_0 \in \{-1,1\}$. Since the zeros of
    $b \in \cB_{even} \cap \cB_{real}$ are symmetric with respect to real and imaginary lines, it
    follows that they are uniquely determined by the zeros of $B = bb_*$. Thus, there exists a unique
    $b \in \cB_{even} \cap \cB_{real}$ such that $aa_* - bb_* = 1$ and
    $\xi_0(b) = \xi_0$. By Theorem \ref{t1},
    there exists a unique $q \in \cP_{even} \cap \cP_{real}$ such that $b = b(\cdot,q)$ and
    then we have $a = a(\cdot,q)$.
\end{proof}

\section{Canonical systems} \label{p6}
In this section, we consider the inverse scattering theory for canonical systems given by
\[ \label{p1e7}
    Jy'(x,z) = k \gh(x) y(x,z),\qq (x,k) \in \R \ts \C,\qq J = \ma 0 & 1 \\ -1 & 0 \am,
\]
where $\gh: \R \to \cM^+_2(\R)$ is a Hamiltonian and by $\cM^+_2(\R)$ we denote the set of
$2 \times 2$ positive-definite self-adjoint matrices with real entries.
Under some condition on the Hamiltonian, the canonical system (\ref{p1e7}) corresponds to
an self-adjoint operator
$$
    \cK = \cK(\gh) = \gh^{-1} J \frac{d}{dx}
$$
in the weighted Hilbert space
$L^2(\R, \C^2, \gh)$ equipped with the norm
$$
    \|f\|^2_{L^2(\R, \C^2, \gh)} = \int_{\R} (\gh(x) f(x), f(x)) dx,\qq f \in L^2(\R, \C^2, \gh),
$$
where $(\cdot,\cdot)$ is the standard scalar product in $\C^2$ (see e.g. \cite{R14}).
It is known that the Dirac equation can be written
as a canonical system (see e.g. \cite{R14}). In order to describe
this result, it is convenient to deal with another form of the Dirac
operator. Recall that $H = H(q)$ is the Dirac operator given by
(\ref{intro:operator}) for some $q \in \cL$. The unitary
transformation of operator $H(q)$ with the unitary operator $T = \frac{1}{\sqrt{2}} \left( \begin{smallmatrix} i
& -i \\ 1 & 1 \end{smallmatrix} \right)$ gives the Dirac
operator
\[ \label{p1e8}
    H_D(q) = T^* H(q) T = J \frac{d}{dx} + V_q,
\]
where
$$
    V_q = \ma q_1 & q_2 \\ q_2 & -q_1 \am,\qq q = -q_2 + i q_1.
$$
We have a similar relation between solutions of the Dirac
equations. Let $v(x,k)$ be a matrix-valued solution of the equation
$$
    v'(x,k) = Q(x) v(x,k) + ik\s_3 v(x,k)\qq (x,k) \in \R \ts \C,
$$
where $Q$ is given by (\ref{intro:potential}) for some $q \in \cL$. Then $u(x,k) = Tv(x,k)$
is a matrix-valued solution of the equation
\[ \label{p1e4}
    J u'(x,k) + V_q(x) u(x,k) = k u(x,k),\qq (x,k) \in \R \ts \C.
\]
We introduce a fundamental $2 \times 2$ matrix-valued solution $M(x,k,q)$ of equation
(\ref{p1e4}) with potential $V_q$ satisfying the initial condition $M(0,k,q) = I_2$,
where $I_2$ is the $2 \times 2$ identity matrix.
Let $r(x,q) = M(x,0,q)$, $x \in \R$, and let $y(x,k) = r^{-1}(x,q) M(x,k,q)$,
$(x,k) \in \R \ts \C$. Then $y(x,k)$ is a solution of the canonical system (\ref{p1e7}) with the
Hamiltonian $\gh_{q} = r^*(\cdot,q) r(\cdot,q)$. Moreover, the associated operators are unitary equivalent
(see Theorem \ref{t7}).

Now, we introduce a class of Hamiltonians associated with the Dirac operators.
By $\cM_2(\R)$ we denote the set of $2 \times 2$ matrices with real entries.
\begin{definition*}
    $\cG_{\bu}^o$ is the set of all function $\gh:\R \to \cM^+_2(\R)$ such that
    $$
        \gh' \in L^2(\R,\cM_2(\R)) \cap L^1(\R,\cM_2(\R)),\qq \gh(0) = I_2,\qq \det \gh(x) = 1,\qq x \in \R.
    $$
\end{definition*}
We also introduce a class of Hamiltonians, which is non-constant only on a finite interval.
\begin{definition*}
    Let $\g > 0$. Then $\cG_{\g}^o \ss \cG_{\bu}^o$ is the subset of all function
    $\gh:\R \to \cM^+_2(\R)$ such that
    $$
        \inf \supp \gh' = 0,\qq \sup \supp \gh' = \g.
    $$
\end{definition*}
\begin{remark}
    Let $\gh \in \cG_{\bu}^o$. Due to $\det \gh(x) = 1$ for any $x \in \R$, we have
    \[ \label{p1e12}
        \gh = \ma \gp & \gq \\ \gq & \frac{1+\gq^2}{\gp} \am,
    \]
    where $\gp : \R \to \R$ and $\gq : \R \to \R$ satisfy
    $$
        \gp', \gq', \left( \frac{1+\gq^2}{\gp} \right)' \in \cL,\qq \gp(0) = 1,\qq \gq(0) = 0,\qq \gp(x) > 0,\qq x \in \R.
    $$
    Moreover, if $\gh \in \cG_{\g}^o$, then $\gh(x)$ is a constant matrix from $\cM^+_2(\R)$ for any $x \in \R \sm [0,\g]$.
\end{remark}
Now, we show that a canonical system $\cK(\gh)$, where $\gh \in \cG_{\bu}^o$, is unitary equivalent to a
Dirac operator $H_D(q)$ for some $q \in \cL$.
\begin{theorem} \label{t7}
    The mapping $q \mapsto \gh_q = r^*(\cdot,q)r(\cdot,q)$ is a bijection between $\cL$ and $\cG_{\bu}^o$
    and its restriction on $\cP_{\g}$ is a bijection between $\cP_{\g}$ and $\cG_{\g}^o$ for any $\g > 0$.
    Moreover, if $\gh_q$ has form (\ref{p1e12}) and $q = -q_2 + i q_1$, then we have
    \[ \label{p1e13}
        \begin{aligned}
            q_1 &= -\frac{1}{2}(\gm \cos \gt + \gn \sin \gt),\qq q_2 = \frac{1}{2}(\gn \cos \gt - \gm \sin \gt),\\
            \gn &= \frac{\gp'}{\gp},\qq \gm = \frac{\gp\gq' - \gp'\gq}{\gp},\qq \gt(x) = \int_0^x \gm(\t) d\t,\qq x \in \R.
        \end{aligned}
    \]
    Furthermore, the associated operators are unitary equivalent and satisfy
    \[ \label{p1e11}
        \cK(\gh_q) = U_q^* H_D(q) U_q,
    \]
    where $U_q: L^2(\R, \C^2, \gh_q) \to L^2(\R,\C^2)$ is a unitary operator given by
    $$
        U_q f = r(\cdot,q)f,\qq f \in L^2(\R, \C^2, \gh_q).
    $$
\end{theorem}
\begin{remark}
    Due to (\ref{p1e13}), if $\gq = 0$, i.e. $\gh$ is a diagonal matrix, then we have $\gm = \gt = 0$,
    which yields $q_1 = 0$ and $q_2 = \frac{\gp'}{2\gp}$. In this case, $H_D(q)$ is
    a supersymmetric Dirac operator and its square is a direct sum of two Schr{\"o}dinger operators with
    singular potentials (see e.g. \cite{T92}).
\end{remark}
\begin{proof}
    It has been proved in Theorem 1.7 from \cite{KM20} that the mapping
    $q \mapsto \cH_q = r^*(\cdot,q) r(\cdot,q)$ is a bijection between $\cP_{\g}$ and
    $\cG_{\g}^o$. Moreover, in the proof of Theorem 1.7 in \cite{KM20}, it has been shown that
    this mapping is an injection and we have (\ref{p1e13}) for $q \in \cL$.
    Now, we show that $q$ given by (\ref{p1e13}) belongs to $\cL$. Since
    $\gh' \in L^2(\R,\cM_2(\R)) \cap L^1(\R,\cM_2(\R))$, we have
    $\gp', \gq' \in \cL$ and $\Gg = \left( \frac{1+\gq^2}{\gp} \right)' \in \cL$.
    By direct calculations, we obtain
    $$
        \Gg = \left( \frac{1+\gq^2}{\gp} \right)' = \frac{2\gq \gq'}{\gp} - \frac{1 + \gq^2}{\gp} \frac{\gp'}{\gp},
    $$
    which yields
    \[ \label{p8e5}
        \frac{\gp'}{\gp} = \frac{2\gq}{1 + \gq^2} \gq' - \frac{\gp}{1 + \gq^2} \Gg.
    \]
    Using $\gp', \gq' \in \cL$, we get $\gp, \gq \in L^{\iy}(\R)$ and then we obtain
    \[ \label{p8e6}
        \frac{2\gq}{1 + \gq^2}, \frac{\gp}{1 + \gq^2} \in L^{\iy}(\R).
    \]
    Using (\ref{p8e5}) and (\ref{p8e6}), we see that $\frac{\gp'}{\gp} \in \cL$ and then we have
    $\gn, \gm \in \cL$. Since $\gt$ is real-valued, we have $|\sin \gt| \leq 1$ and $|\cos \gt| \leq 1$.
    Hence, it follows from (\ref{p1e13}) that $q_1,q_2 \in \cL$.

    Now, we prove that $U_q$ is unitary and (\ref{p1e11}) holds true. Let $f \in L^2(\R,\C^2,\gh_q)$
    and $g \in L^2(\R,\C^2)$. Then we have
    $$
        \begin{aligned}
            (U_qf,g)_{L^2(\R,\C^2)} &= \int_{\R} (r(x)f(x), g(x)) dx = \int_{\R} ((r^{-1})^*(x)r^*(x)r(x)f(x), g(x)) dx\\
            &= \int_{\R} (r^*(x)r(x)f(x),r^{-1}(x) g(x)) dx = (f,U_q^*g)_{L^2(\R,\C^2,\gh_q)},
        \end{aligned}
    $$
    where $(U_q^*g)(x) = r^{-1}(x)g(x)$, $x \in \R$. Thus, $U_qU_q^*$ is the identity operator in $L^2(\R,\C^2)$
    and $U_q^*U_q$ is the identity operator in $L^2(\R,\C^2,\gh_q)$, which yields that $U_q$ is unitary.
    Let $f$ be a differentiable function. Then we have
    \[ \label{p8e7}
        r^{-1}(J(rf)'+V_qrf) = r^{-1}(Jr'f + Jrf' + V_qrf) = r^{-1}Jrf' = \gh_q^{-1}Jf',
    \]
    where we used $r^{-1}Jr = \gh_q^{-1}J$, which can be proved by direct calculation.

    Finally, we show that $\cD(H_D(q)) = U_q \cD(\cK(\gh_q))$ for any $q \in \cL$, where $\cD(\cdot)$
    denotes the domain of the operator. Using the identity $y(x,k,\gh_q) = r(x,q) M(x,k,q)$,
    $(x,k) \in \R \ts \C$, where $q \in \cL$, $y$ is a solution of the canonical system and
    $M$ is a solution of the Dirac equation, we see that
    $y(\cdot,k,\gh_q) \in L^2(\R_{\pm},\C^2,\gh_q)$ if and only if $M(\cdot,k,q) \in L^2(\R,\C^2)$.
    Thus, a Dirac operator and associated canonical system are in the limit-point case at $\pm \iy$ simultaneously.
    It is known that $H_{D}(q)$ is in the limit-point case at $\pm \iy$ for any $q \in \cL$
    (see e.g. Theorem 6.8 in \cite{Weid87}). Hence, these operators are self-adjoint on the maximal
    domains
    $$
        \begin{aligned}
            \cD(\cK(\gh_q)) &= \{ f \in L^2(\R,\C^2,\gh_q) \mid \text{$f$ is absolutely continuous},\,
            \gh_q^{-1}Jf \in L^2(\R,\C^2,\gh_q) \},\\
            \cD(H_{D}(q)) &= \{ f \in L^2(\R,\C^2) \mid \text{$f$ is absolutely continuous},\,
            Jf'+V_qf \in L^2(\R,\C^2) \}.
        \end{aligned}
    $$
    Using (\ref{p8e7}), we get
    $$
        \| J(U_qf)'+V_q(U_qf) \|_{L^2(\R,\C^2)} =
        \int_{\R} (r(x)\gh_q^{-1}(x)Jf'(x),r(x)\gh_q^{-1}(x)Jf'(x)) dx =
        \| \gh_q^{-1}Jf' \|_{L^2(\R,\C^2,\gh_q)}.
    $$
    It follows that $\gh_q^{-1}Jf' \in L^2(\R,\C^2,\gh_q)$ if and only if
    $J(U_qf)'+V_q(U_qf) \in L^2(\R,\C^2)$.
    Hence $\cD(H_D(q)) = U_q \cD(\cK(\gh_q))$ for any $q \in \cL$ and (\ref{p1e11}) holds true.
\end{proof}
Combining (\ref{p1e8}) and (\ref{p1e11}), we obtain
\[ \label{p1e14}
    \cK(\gh_q) = (T U_q)^* H(q) (T U_q),
\]
i.e. these operators are unitary equivalent for any $q \in \cL$. The scattering theory for the
operator $H(q)$ is well-known and it was described in Section \ref{p2}.
Now, we introduce the scattering matrix for the operator
$\cK(\gh)$, where $\gh \in \cG_{\bu}$.
Let $\cK_o$ be a self-adjoint operator in $L^2(\R,\C^2)$ given by $\cK_o = J \frac{d}{dx}$. Note that
$\cK_o = H_D(0) = \cK(I_2)$, where $I_2$ is the $2 \ts 2$ identity matrix and then it can be
considered as a free Dirac operator and a free canonical system. We introduce the wave
operators for the pair $\cK(\gh)$ and $\cK_o$ with bounded linear identification operator
$\cJ : L^2(\R,\C^2) \to L^2(\R,\C^2,\gh)$ as follows
\[ \label{p8e9}
    W_{\pm}(\cK(\gh),\cK_o,\cJ) = \slim_{t \to \pm \iy} e^{it\cK(\gh)} \cJ e^{-it\cK_o}.
\]
If the wave operators exist, then we introduce the scattering operator
\[ \label{p8e10}
    S(\cK(\gh), \cK_o, \cJ) = W_{+}^*(\cK(\gh),\cK_o,\cJ) W_{-}(\cK(\gh),\cK_o,\cJ).
\]
Let $\cK_o$ be diagonalized by an unitary operator $\cF_o$, i.e. $\cF_o^* \cK_o \cF_o$ acts as
multiplication by independent variable in $L^2(\R,\C^2)$. Then we introduce the scattering matrix
by
\[ \label{p8e11}
    S(k,\cK(\gh), \cK_o, \cJ) = \cF_o^* S(\cK(\gh), \cK_o, \cJ) \cF_o,\qq k \in \s_{ac}(\cK(\gh)) = \R.
\]
Similarly, for any $q \in \cL$, we can introduce the wave operators $W_{\pm}(H(q),H_o)$, the scattering operator $S(H(q),H_o)$
and the scattering matrix $S(z,H(q),H_o)$, where the identification operator is the identity operator
on $L^2(\R,\C^2)$. It is well known that for any $q \in \cL$ the wave operators $W_{\pm}(H(q),H_o)$
exist and complete. Recall that we have introduced the scattering matrix $S(k,H(q),H_o) = S(k,q)$
in (\ref{p2e3}). Due to (\ref{p1e14}), we get the following corollary.
\begin{corollary}
    Let $q \in \cL$. Then the wave operators $W_{\pm}(\cK(\gh_q),\cK_o,U_{q}^*)$
    exist and complete and the following identity holds true:
    \[ \label{p8e2}
        S(k,\cK(\gh_{q}), \cK_o, U_{q}^*) = S(k,q),\qq k \in \R,
    \]
    where $S(k,q)$ is given by (\ref{p2e3}).
\end{corollary}
\begin{proof}
    Let $q \in \cL$. Then, using (\ref{p1e14}), we obtain
    \[ \label{p8e8}
        e^{it\cK(\gh_q)} = (T U_q)^* e^{itH(q)} (T U_q),\qq e^{it\cK_o} = T^* e^{itH_o} T,\qq t \in \R.
    \]
    Substituting (\ref{p8e8}) in (\ref{p8e9}), we get
    $$
        \begin{aligned}
            W_{\pm}(\cK(\gh_q),\cK_o,U_{q}^*) &= \slim_{t \to \pm \iy} e^{it\cK(\gh_q)} U_{q}^* e^{-it\cK_o}
            =\slim_{t \to \pm \iy} (T U_q)^* e^{itH(q)} T U_q U_{q}^* T^* e^{itH_o} T\\
            &= (T U_q)^* W_{\pm}(H(q),H_o) T.
        \end{aligned}
    $$
    Since $T$ and $U_q$ are unitary operators and the wave operators $W_{\pm}(H(q),H_o)$ exist and
    complete, we get that the wave operators $W_{\pm}(\cK(\gh_q),\cK_o,U_{q}^*)$ exist and complete.
    Using (\ref{p8e10}), we get
    \[ \label{p8e12}
        S(\cK(\gh_{q}), \cK_o, U_{q}^*) = T^* S(H(q),H_o) T.
    \]
    It follows from (\ref{p1e14}) that $H_o$ is diagonalized by the unitary operator $T \cF_o$,
    i.e. $(T\cF_o)^* H_o T \cF_o$ acts as multiplication by independent variable in $L^2(\R,\C^2)$.
    Using (\ref{p8e11}) and (\ref{p8e12}), we get (\ref{p8e2}).
\end{proof}
Due to (\ref{p2e3}) and (\ref{p8e2}), the scattering matrix $S(k,\cK(\gh_{q}), \cK_o, U_{q}^*)$
has the following form for any $q \in \cL$:
\[ \label{p8e13}
    S(k,\cK(\gh_q), \cK_o, U^*) = \ma \frac{1}{\ga} & \gr_+ \\ \gr_- & \frac{1}{\ga} \am(k,\gh_q),
    \qq k \in \R,
\]
where $\gr_{+} = -\frac{\ol \gb}{\ga}$ and $\gr_- = \frac{\gb}{\ga}$. Thus, using obtained results
about inverse scattering for the Dirac operators, we construct the inverse scattering for
the canonical systems.
Recall that the classes $\cR_{\g}$ and $\cB_{\g}$ were defined in Section \ref{p1} and
the classes $\cR_{\bu}$ and $\cB_{\bu}$ were defined in Section \ref{p2}.
\begin{corollary} \label{t10}
    \begin{enumerate}[label={\roman*)}]
        \item The mappings $\gh \mapsto \gr_{\pm}(\cdot,\gh)$ are bijections between $\cG_{\bu}^o$ and $\cR_{\bu}$ and
        theirs restrictions on $\cG_{\g}^o$ is a bijection between $\cG_{\g}^o$ and $\cR_{\g}$ for any $\g > 0$.
        \item The mapping $\gh \mapsto \gb(\cdot,\gh)$ is a bijection between $\cG_{\bu}^o$ and $\cB_{\bu}$ and
        its restriction on $\cG_{\g}^o$ is a bijection between $\cG_{\g}^o$ and $\cB_{\g}$ for any $\g > 0$.
    \end{enumerate}
\end{corollary}
\begin{proof}
    Using (\ref{p2e3}), (\ref{p8e2}), and (\ref{p8e13}), we have
    $$
        \ga(\cdot,\gh_q) = a(\cdot,q),\qq \gb(\cdot,\gh_q) = b(\cdot,q),\qq \gr_{\pm}(\cdot,\gh_q) = r_{\pm}(\cdot,q)
    $$
    for any $q \in \cL$. By Theorem \ref{t7}, the mapping $q \mapsto \gh_q$ is a bijection between $\cL$ and $\cG_{\bu}^o$
    and its restriction on $\cP_{\g}$ is a bijection between $\cP_{\g}$ and $\cG_{\g}^o$ for any $\g$.
    Thus, the statement of the corollary follows from Theorems \ref{t1}, \ref{t11}, \ref{hlt1}, and \ref{hlt2}.
\end{proof}

Note that the Hamiltonians $\gh \in \cG_{\bu}^o$ are normed in such way that $\gh(0) = I_2$ and $\det \gh(x) = 1$,
$x \in \R$. Now, we consider a more general classes of the Hamiltonians.
\begin{definition*}
    $\cG_{\bu}$ is the set of all function $\gh:\R \to \cM^+_2(\R)$ such that
    \[ \label{p8e3}
        \r(x) = |\det \gh(x)|^{1/2},\qq \vt(x) = \int_0^x \r(x)dx,\qq
        \gh_1(x) = \frac{\gh(\vt^{-1}(x))}{\r(\vt^{-1}(x))},\qq x \in \R,
    \]
    where $\vt^{-1}$ is the inverse function to $\vt$, satisfy
    \begin{enumerate}[label={(\roman*)}]
        \item $\r \in \cL^1_+(\R)$, where $\cL^1_+(\R)$ is a class of functions $f$ such that
        $f \in L^1_{loc}(\R)$, $f > 0$ and $\int_0^x f(t) dt \to \pm \iy$ as $x \to \pm \iy$.
        \item $\gh_1$ is absolutely continuous and $\gh_1' \in L^2(\R,\cM_2(\R))$.
    \end{enumerate}
\end{definition*}
We also introduce a class of Hamiltonians, which is non-constant only on a finite interval.
\begin{definition*}
    Let $\g > 0$. Then $\cG_{\g} \ss \cG_{\bu}$ is the subset of all function
    $\gh:\R \to \cM^+_2(\R)$ such that
    $$
        \inf \supp \gh_1' = 0,\qq \sup \supp \gh_1' = \g,
    $$
    where $\gh_1$ is given by (\ref{p8e3}).
\end{definition*}
By $\cU_2^+(\R)$, we denote a class of $2 \times 2$ upper-triangle positive real-valued matrix $C$
such that $\det C = 1$.
\begin{theorem} \label{t13}
    The mapping $(\gh_o,C,\r) \mapsto \gh = \r C^* (\gh_o \circ \vt) C$, where $\vt(x) = \int_0^x \r(t) dt$,
    is a bijection between $\cG_{\bu}^o \ts \cU_2^+(\R) \ts \cL^1_+(\R)$ and $\cG_{\bu}$ and its restriction on
    $\cG_{\g}^o \ts \cU_2^+(\R) \ts \cL^1_+(\R)$ is a bijection between
    $\cG_{\g}^o \ts \cU_2^+(\R) \ts \cL^1_+(\R)$ and $\cG_{\g}$
    for any $\g > 0$. Moreover, the following identities hold true:
    $$
        \r^2 = \det \gh,\qq C^*C = \frac{\gh(0)}{|\det \gh(0)|^{1/2}}.
    $$
    Furthermore, the corresponding operators are unitary equivalent and satisfy
    \[ \label{p8e15}
        \cK(\gh) = V^* \cK(\gh_o) V,
    \]
    where $V: L^2(\R,\C^2,\gh) \to L^2(\R,\C^2,\gh_o)$ is a unitary operator given by
    $$
        (Vf)(x) = Cf(\vt^{-1}(x)),\qq x \in \R,
    $$
    where $C \in \cM_2^u$ such that $C^*C = \gh(0)$ and $\vt(x) = \int_0^x |\det \gh(t)|^{1/2} dt$, $x \in \R$.
\end{theorem}
\begin{remark}
    Note that for any $f \in L^2(\R,\C^2,\gh_o)$, we get
    $$
        (V^*f)(x) = C^{-1}f(\vt(x)),\qq x \in \R.
    $$
\end{remark}
\begin{proof}
    Firstly, we show that $\gh = \r C^* (\gh_o \circ \vt) C \in \cG_{\bu}$ for any
    $(\gh_o,C,\r) \in \cG_{\bu}^o \ts \cU_2^+(\R) \ts \cL^1_+(\R)$. By direct calculations, we have
    $$
        |\det \gh(x)|^{1/2} = \r(x),\qq \vt(x) = \int_0^x \r(x)dx,\qq
        \frac{\gh(\vt^{-1}(x))}{\r(\vt^{-1}(x))} = C^*\gh_o(x)C,\qq x \in \R.
    $$
    Thus, it follows from the definition of $\cG_{\bu}$ that $\gh \in \cG_{\bu}$.
    Moreover, if $\gh_o \in \cG_{\g}^o$, then we have
    $$
        \inf \supp \gh_o' = 0,\qq \sup \supp \gh_o' = \g,
    $$
    which yields that $\gh \in \cG_{\g}$.

    Secondly, we show that this mapping is an injection. Let
    $$
        \r_1(x) C_1^* \gh^1_o(\vt_1(x)) C_1 = \r_2(x) C_2^* \gh^2_o(\vt_2(x)) C_2
    $$
    for almost all $x \in \R$, where
    $(\gh^j,C_j,\r_j) \in \cG_{\bu}^o \ts \cU_2^+(\R) \ts \cL^1_+(\R)$ and
    $\vt_j(x) = \int_0^x \r_j(t) dt$ for $j = 1,2$.
    Since $\det \gh^j_o = 1$ and $\det C_j = 1$ for any $j = 1,2$, it follows that
    $\r_1 = \r_2$ and then we have
    \[ \label{p8e14}
        \vt_1(x) = \vt_2(x),\qq C_1^* \gh^1_o(x) C_1 = C_2^* \gh^2_o(x) C_2,\qq x \in \R.
    \]
    Substituting $x = 0$, we get $C_1^* C_1 = C_2^* C_2$. It is known that
    every $A \in \cM_2^+(\R)$ has a unique Cholesky factorization, i.e.
    there exists a unique $C \in \cU_2^+(\R)$ such that $A = C^* C$.
    Thus, it follows that $C_1 = C_2$ and then, due to (\ref{p8e14}), we have
    $\gh^1_o = \gh^2_o$.

    Thirdly, we show that this mapping is a surjection. Let $\gh \in \cG_{\bu}$ and let
    $$
        \r(x) = |\det \gh(x)|^{1/2},\qq \vt(x) = \int_0^x \r(x)dx,\qq
        \gh_1(x) = \frac{\gh(\vt^{-1}(x))}{\r(\vt^{-1}(x))},\qq x \in \R.
    $$
    It follows from the definition of $\cG_{\bu}$ that
    $\r \in \cL^1_+(\R)$ and $\gh_1(0) \in \cM_2^+(\R)$. As it was mentioned above,
    there exists a unique $C \in \cU_2^+(\R)$ such that $\gh_1(0) = C^* C$. We introduce $\gh_o = (C^{-1})^* \gh_1
    C^{-1}$. Then we have $\gh = \r C^* (\gh_o \circ \vt) C$ and $\gh_o \in \cG_{\bu}^o$.
    Moreover, if $\gh \in \cG_{\g}$, then we have
    $$
        \inf \supp \gh_1' = 0,\qq \sup \supp \gh_1' = \g,
    $$
    which yields $\gh_o \in \cG_{\g}^o$.

    Now, we prove that $V$ is unitary and (\ref{p8e15}) holds true.
    Let $f \in L^2(\R,\C^2,\gh)$ and $g \in L^2(\R,\C^2,\gh_o)$. We introduce
    $y = \vt^{-1}(x)$. Then we have $x = \vt(y)$, $dx = \r(y) dy$, and
    $$
        \begin{aligned}
            (Vf,g)_{L^2(\R,\C^2,\gh_o)} &= \int_{\R} (\gh_o(x) Cf(\vt^{-1}(x)), g(x)) dx
            = \int_{\R} (C^*\gh_o(\vt(y)) f(y), g(\vt(y))) \r(y) dy\\
            &= \int_{\R} (\r(y)C^*\gh_o(\vt(y))C f(y),C^{-1} g(\vt(y))) dy
            = (f,V^*g)_{L^2(\R,\C^2,\gh)},
        \end{aligned}
    $$
    where $(V^*g)(y) = C^{-1}g(\vt(y))$, $y \in \R$. It is easy to see that
    $V V^*$ is the identity operator in $L^2(\R,\C^2,\gh_o)$ and $V^* V$ is the identity operator
    in $L^2(\R,\C^2,\gh)$, which yields that $V$ is unitary. Moreover, by direct calculation, we obtain
    $$
        (V^* \cK(\gh_o) V f)(x) = (\cK(\gh)f)(x),\qq x \in \R,
    $$
    for any differentiable function $f:\R \to \C^2$.

    Finally, we show that $\cD(\cK(\gh_o)) = V \cD(\cK(\gh))$, where $\cD(\cdot)$
    denotes the domain of the operator.
    Recall that operators $\cK(\gh)$ are self-adjoint on
    the maximal domains
    $$
        \cD(\cK(\gh)) = \{ f \in L^2(\R,\C^2,\gh) \mid \text{$f$ is absolutely continuous},\,
        \gh^{-1}Jf \in L^2(\R,\C^2,\gh) \}.
    $$
    Using $(Vf)'(x) = C f'(\vt^{-1}(x))\frac{1}{\r(\vt^{-1}(x))}$, $x \in \R$, we get
    $$
        \| \gh_o^{-1} J (Vf)' \|_{L^2(\R,\C^2,\gh_o)} =
            \int_{\R} \frac{1}{\r^2(\vt^{-1}(x))}(\gh_o(x)\gh_o^{-1}(x) J Cf'(\vt^{-1}(x)),\gh_o^{-1}(x) J Cf'(\vt^{-1}(x))) dx.
    $$
    Substituting $y = \vt^{-1}(x)$ and using $C^* J C = J$, we obtain
    $$
        \begin{aligned}
            \| \gh_o^{-1} J (Vf)' \|_{L^2(\R,\C^2,\gh_o)} &=
            \int_{\R} \frac{1}{\r(y)}(\gh_o(\vt(y))\gh_o^{-1}(\vt(y)) J Cf'(y),\gh_o^{-1}(\vt(y)) J Cf'(y)) dy\\
            &= \int_{\R} ((C^{-1})^*\gh(y)\gh^{-1}(y) C^* J C f'(y),C\gh^{-1}(y) C^*J Cf'(y)) dy\\
            &= \int_{\R} (\gh(y)\gh^{-1}(y) J f'(y),\gh^{-1}(y) J f'(y)) dy
            = \| \gh^{-1}Jf' \|_{L^2(\R,\C^2,\gh)}.
        \end{aligned}
    $$
    It follows from this identity that $\gh^{-1}Jf' \in L^2(\R,\C^2,\gh)$ if and only if
    $\gh_o^{-1}J(Vf)' \in L^2(\R,\C^2,\gh_o)$.
    Hence, we have $\cD(\cK(\gh_o)) = V \cD(\cK(\gh))$ and (\ref{p8e15}) holds true.
\end{proof}
Since $\cK(\gh)$ and $\cK(\gh_o)$ are unitary equivalent, we can construct the scattering theory
for the operators $\cK(\gh)$, $\gh \in \cG_{\bu}$, as above.
\begin{theorem}
    Let $\gh \in \cG_{\bu}$ and let $\gh_o$ be given by (\ref{p8e3}).
    Then the wave operators\\
    $W_{\pm}(\cK(\gh),\cK_o,(UV)^*)$ exist and complete and the following identity holds true:
    \[ \label{p8e18}
        S(k,\cK(\gh),\cK_o,(UV)^*) = S(k,\cK(\gh_o),\cK_o,U^*) = S(k,q),\qq k \in \R,
    \]
    where $q \in \cL$ is such that $\gh_q = \gh_o$ and $S(k,q)$ is given by (\ref{p2e3}).
\end{theorem}
\begin{remark}
    It follows from Theorems \ref{t10} and \ref{t13} and identity (\ref{p8e18})
    that a Hamiltonian $\gh \in \cG_{\bu}$ is uniquely determined by $\gh(0)$, $\det \gh$ and
    by its scattering matrix $S(k,\cK(\gh),\cK_o,(UV)^*)$, $k \in \R$.
\end{remark}
\begin{proof}
    Using (\ref{p8e15}), we get $e^{it\cK(\gh)} = (V)^* e^{it\cK(\gh_o)} V$ for any $t \in \R$,
    which yields
    \[ \label{p8e19}
        \begin{aligned}
            W_{\pm}(\cK(\gh),\cK_o,(UV)^*) &= \slim_{t \to \pm \iy} e^{it\cK(\gh)} (UV)^* e^{-it\cK_o}
            =\slim_{t \to \pm \iy} V^* e^{it\cK(\gh_o)} U^* e^{-it\cK_o}\\
            &= V^* W_{\pm}(\cK(\gh_o),\cK_o, U^*).
        \end{aligned}
    \]
    Since $V$ is a unitary operator, the wave operators $W_{\pm}(\cK(\gh),\cK_o,(UV)^*)$
    exists and complete. Moreover, using (\ref{p8e2}) and (\ref{p8e19}), we obtain (\ref{p8e18}).
\end{proof}

\footnotesize

\no {\bf Acknowledgments.} E.K. and D. M. are supported by the RFBR grant No. 19-01-00094.

\end{document}